\documentclass[reqno, a4paper, 11pt]{amsart}
 \usepackage[foot]{amsaddr}
\usepackage{amssymb,amsmath,amsthm, mathrsfs}
\usepackage{graphics}
\usepackage{hyperref}
\usepackage{bbm}
\usepackage[shortlabels]{enumitem}         
\usepackage{todonotes}
\usepackage{url}
\usepackage{textcomp}
\usepackage{booktabs}
\usepackage{mathtools}
\usepackage{adjustbox}
\usepackage{rotating}
\usepackage{multirow}
\usepackage{commath}
\usepackage[capitalize]{cleveref}
\usepackage{float}
\usepackage{tikz}
\usepackage[font=small]{caption}
\usepackage{subcaption}
\newtheorem{thm}{Theorem}[section]
\newtheorem{cor}[thm]{Corollary}
\newtheorem{lem}[thm]{Lemma}

\theoremstyle{definition}
\newtheorem{defn}[thm]{Definition}
\theoremstyle{remark}
\newtheorem{rem}[thm]{Remark}

\numberwithin{equation}{section}

\newcommand{\RR}{\mathbb{R}}

\newcommand{\QQ}{\mathbb{Q}}

\newcommand{\NN}{\mathbb{N}}

\newcommand{\cA}{\mathcal{A}}

\newcommand{\cF}{\mathcal{F}}
\newcommand{\cG}{\mathcal{G}}

\newcommand{\pp}{\textup{\texttt{+}}}
\newcommand{\mm}{\textup{\texttt{-}}}

\newcommand{\pOne}{\textup{\texttt{+1}}}
\newcommand{\mOne}{\textup{\texttt{-1}}}
\newcommand{\pmOne}{\textup{\texttt{\textpm 1}}}
\newcommand{\tL}{\textup{\texttt{L}}}
\newcommand{\tR}{\textup{\texttt{R}}}
\newcommand{\wind}{\textrm{wind}}  

\newcommand{\Ind}[1]{\mathbf{1}_{\left\{#1\right\}}}

\newcommand{\Excond}[3]{\mathbb{E}^{#1}\left[\left.#2\right|#3\right]}  

\title[Term structure shapes in the two-factor Vasicek model]{State space decomposition and classification of term structure shapes in the two-factor Vasicek model}
\author{Martin Keller-Ressel, Felix Sachse}
\address{Department of Mathematics, TU Dresden, Germany}
\email{martin.keller-ressel@tu-dresden.de, felix.sachse@tu-dresden.de}

\begin{document}
\maketitle

\begin{abstract}
    Using the concept of envelopes we show how to divide the state space $\RR^2$ of the two-factor Vasicek model into regions of identical term-structure shape. We develop a formula for determining the shapes utilizing winding numbers and give a nearly complete classification of the parameter space regarding the occurring shapes.  
\end{abstract}
\tableofcontents
\newpage

\section{Introduction}\label{sec:intro}
The two-dimensional Vasicek model (see \cite{brigo2007interest}) is a stochastic model for the evolution of the short-term interest rate. From such a model, the whole term structure of forward rates and bond yields can be derived in terms of risk-neutral expectations. It extends the single-factor Vasicek model \cite{vasicek1977equilibrium} and uses an Ornstein-Uhlenbeck process evolving in the state space $\mathcal{Z} = \RR^2$ as driver of the short rate. An important motivation to extend the one-factor Vasicek model by an additional factor is that the two-factor model provides a richer set of attainable term structure shapes, see \cite{diez2020yield}. For example, it can produce dipped yield and forward curves, which is impossible in the single-factor case, cf.~\cite{diez2020yield}. As it can be seen in \cref{fig:euro}, market-implied forward and yield curves frequently follow even more complex shapes with multiple `humps' and `dips' that are far beyond the scope of time-homogeneous affine single-factor models. Here, we want to answer the apparently simple question: 
\begin{quote}`Which shape of the yield (or forward) curve is produced by the two-dimensional Vasicek model, conditional on a state $Z_t = z$'?
\end{quote}
In other words, we seek to partition the state space $\mathcal{Z}$ into regions $R_1, \dotsc, R_k$ such that each region can be associated to a certain shape of the yield (or forward) curve, such as \texttt{normal}, \texttt{dipped}, \texttt{humped}, \texttt{inverse} and more complex shapes with multiple local extrema. This simple question has a surprisingly complex geometric answer, with conditional shapes and associated regions described by a single curve $\eta$, the \emph{envelope}, which may display singularities and self-intersections. More precisely, the state-contingent shape given state $z \in \mathcal{Z}$ is shown to depend on the \emph{winding number} of the `augmented' envelope $\hat \eta$ around $z$, see \cref{thm:main}.\\
Being able to answer our initial question also enables us to answer several related questions of practical interest: For example, we can determine the probability of each shape to be observed under both the physical and the risk-neutral measure, and we can classify all possible transitions between different yield curve shapes. Moreover, we rederive and refine several results of \cite{keller2021classification} on the scope of shapes that are attainable over the whole state space $\mathcal{Z}$.

        \begin{figure}[bp]
            \centering
            \includegraphics[width=1\textwidth]{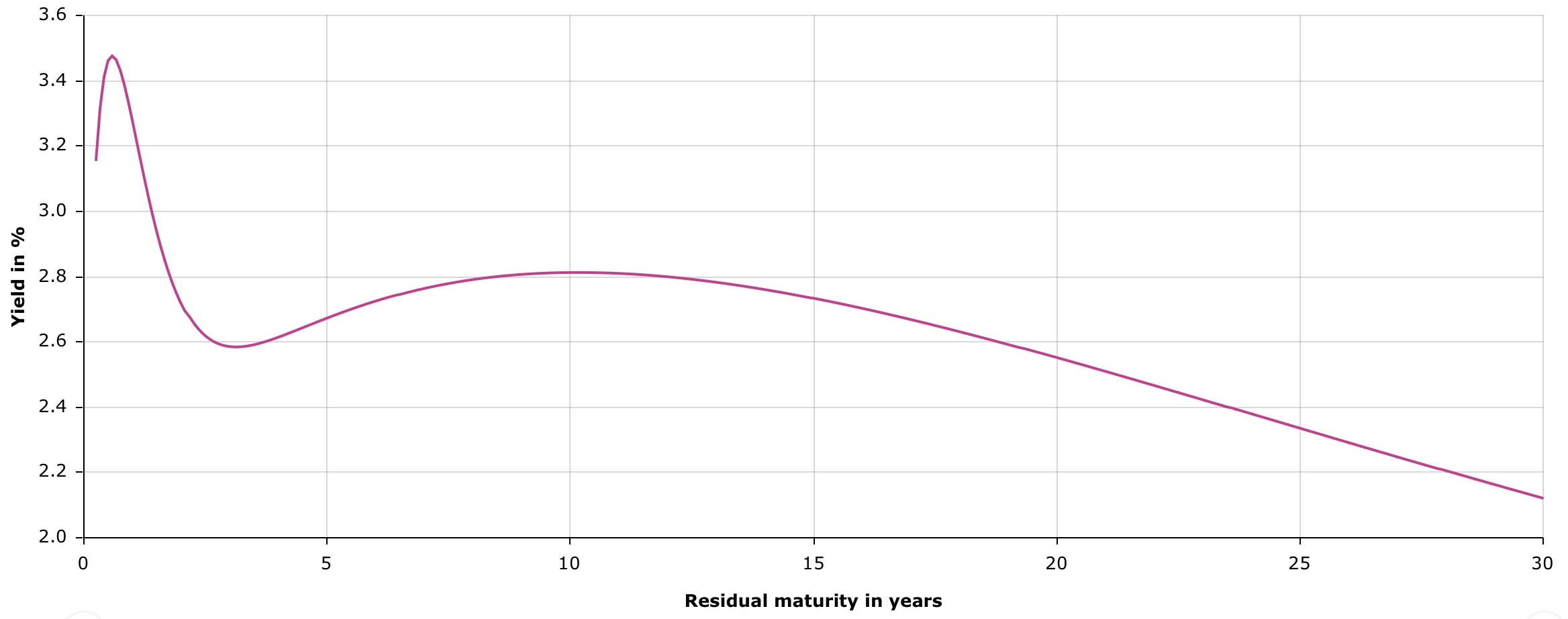}
            \caption{Euro area forward curve, as published by the European Central Bank on March 2nd, 2023. The shape of the forward curve features two humps (local maxima) and one dip (local minimum).\label{fig:euro}}
        \end{figure}

\subsection{Related Literature}
For the one-factor case, it has been shown already in \cite{vasicek1977equilibrium}
that the yield curve can only attain 3 different shapes, namely the \texttt{normal}, \texttt{inverse} and \texttt{humped} type. The shape is solely determined by the value of the short-rate with the associated bounds also given in \cite{vasicek1977equilibrium}. The same holds true for the forward curve and in fact for all one-dimensional affine models, see e.g. \cite{keller2008yield, keller2018correction}. More recently the question of attainable shapes in multi-factor models, in particular the two-factor model, has been considered. Motivated by long-term simulation of interest rates \cite{diez2020yield} show that \texttt{dipped} shapes can occur in a two-factor Vaiscek model and provide bounds on their probability. Using the theory of Descartes systems and total positivity \cite{keller2021classification} gives a full classification of all attainable shapes in the two-factor Vasicek model. Depending on the parameter regime (see \cref{def:regime} below) it is shown that 5, 7 or 9 shapes are attainable. However, the analysis of \cite{keller2021classification} is global over the state space, that is, it does provide only very limited knowledge of the influence of the state $Z_t = z$ on the attained shape and does not allow for a segmentation of the state space $\mathcal{Z}$ according to shape.

     \section{Background}\label{sec:bg}
     We introduce the two-factor Vasicek model, as considered in \cite{brigo2007interest, keller2021classification}, which extends the single-factor model of \cite{vasicek1977equilibrium}. Let $Z=(Z^1, Z^2)$ be a factor process with dynamics
    \begin{equation}\label{eq:SDE}
        \dif Z_t^i = -\lambda_i(Z_t^i - \theta_i)\dif t + \sigma_i\dif B_t^i,\quad i\in\{1,2\}
    \end{equation}
    under the risk-neutral measure $\QQ$, with $\dif B_t^1\dif B_t^2 = \rho\dif t$ and $0 < \lambda_1 < \lambda_2; \theta_1,\theta_2\in\RR; \sigma_1,\sigma_2 > 0$ and $\rho\in[-1,1]$. The short rate $r$ is defined\footnote{Possible scaling factors of $Z_t^1$ or $Z_t^2$ can be absorbed into the SDE \eqref{eq:SDE}, such that generality is not restricted by this definition.} by
    \begin{align*}
        r_t = \kappa +Z_t^1 + Z_t^2,
    \end{align*}
    where $\kappa \in \RR$.
The bond price is given by
    \begin{align*}
        P(t,t+x) = \Excond{\QQ}{\exp\left(-\int_t^{t+x} r_s \,ds\right)}{\cF_t} = \exp(A(x) + Z_t^TB(x))
    \end{align*}
    where
        \begin{align*}
        B(x) = \begin{pmatrix} f(\lambda_1, x) \\ f(\lambda_2, x) \end{pmatrix}, \qquad  f(\lambda, x) &= \frac{1-\text{e}^{-\lambda x}}{\lambda},
    \end{align*}
    and
    \[A(x) = - \frac{\sigma_1^2}{2\lambda_1^2} f(2 \lambda_1,x) - \frac{\sigma_2^2}{2\lambda_2^2} f(2 \lambda_2,x) - r f(\lambda_1 + \lambda_2,x) - u_1 f(\lambda_1,x) - u_2 f(\lambda_2,x), \]
    where
    \begin{align}\label{eq:coeff}
r &= \rho\frac{\sigma_1\sigma_2}{\lambda_1\lambda_2}, \qquad u_i = \left(\theta_i - \frac{\sigma_i^2}{\lambda_i^2} - \rho\frac{\sigma_1\sigma_2}{\lambda_1\lambda_2}\right)\notag
    \end{align}
    for $i \in \set{1,2}$.

Finally, the forward curve $f(x,Z_t)$ and the yield curve $y(x,Z_t)$, are given by 
\begin{equation}\label{eq:fy}
    \begin{split}
        f(x,Z_t) &= -\partial_x \log P(t,t+x) = A'(x) + B'(x)^T Z_t\\
        y(x,Z_t) &= -\frac{1}{x}\log P(t,t+x) = \frac{A(x)}{x} + \frac{B(x)^T}{x}Z_t.
    \end{split}
    \end{equation}
    As in \cite{keller2021classification}, we are interested in the global \emph{shape} of both functions, i.e., their sequence and number of local minima and maxima. In common market terminology a minimum is called a `dip' and a maximum a `hump'. A list of common shapes and their conventional names are given in \cref{tab:shape}. As both yield and forward curve are smooth, we can classify their shapes by analyzing the sign sequence of their derivative with respect to $x$. These functions are given by 
\begin{equation}\label{eq:fy_diff}
    \begin{split}
        \frac{\partial}{\partial x}f(x,Z_t) &= A''(x) + B''(x)^T Z_t\\
        &=: a_{\rm f}(x) + b_{\rm f}(x)Z_t^1 + c_{\rm f}(x)Z_t^2\\
        \frac{\partial}{\partial x} y(x,Z_t)
        &= \frac{A'(x)}{x} - \frac{A(x)}{x^2} + \left(\frac{B'(x)}{x} - \frac{B(x)}{x^2}\right)Z_t\\
        &=: a_{\rm y}(x) + b_{\rm y}(x)Z_t^1 + c_{\rm y}(x)Z_t^2.
    \end{split}
    \end{equation}
    We emphasize that both are \emph{affine} functions of the state $Z_t = (Z_t^1, Z_t^2)$, and denote their coefficients (including the constant term) by $(a(x), b(x), c(x))$ in the general case. It is easily verified that $(b(x), c(x)) \neq 0$ for all $x \in (0,\infty)$; hence these coefficients define a family of (non-degenerate) lines
   \begin{equation}
   \ell_x: \quad a(x) + z_1 b(x) + z_2 c(x) = 0
   \end{equation}
   in the $(z_1,z_2)$-plane. The interpretation of these lines is as follows: If $z \in \ell_x$, then the forward curve (or yield curve) has a local extremum at $x$, given the current state $Z_t = z$. It can be easily shown that also the limiting lines $\ell_0$ (as $x \to 0$) and $\ell_\infty$ (as $x \to \infty$) are well-defined. These lines split the $(z_1,z_2)$-plane into half-spaces $\ell_0^\pm$ and $\ell_\infty^\pm$, such that the forward (or yield) curve is initially increasing/decreasing in $\ell_0^+/\ell_0^-$, as well as terminally increasing/decreasing in $\ell_0^+/\ell_0^-$. Indeed, the complete information regarding all attainable shapes of the forward and yield curve is encoded in this family $\cF = (\ell_x)_{x \in [0,\infty]}$ of lines. To illustrate, we show in \cref{fig:Envelope} the family $\cF$ of lines for a certain parameterization of the two-dimensional Vasicek model.\\
\begin{table}[hbtp]
\begin{center}
\begin{tabular}{p{3cm}p{5cm}p{3cm}} 
\toprule
Shape of the term structure & Description & Sign sequence of derivative\\ 
\midrule 
\texttt{normal} & strictly increasing & $\pp$\\
\texttt{inverse} & strictly descreasing & $\mm$\\
\texttt{humped} & single local maximum & $\pp \mm$\\
\texttt{dipped} & single local minimum & $\mm \pp$\\
\texttt{hd} & hump-dip, i.e. local maximum followed by local minimum & $\pp \mm \pp$\\
\texttt{dh}, \texttt{hdh}, etc. & further sequences of multiple `dips'  and `humps'  & $\dotsc$ \\
\bottomrule
\end{tabular}
\end{center}
\caption{Shapes of the term structure; adapted from \cite{keller2021classification}.\label{tab:shape}}
\end{table}

    \begin{figure}
        \centering
        \includegraphics[width=\textwidth]{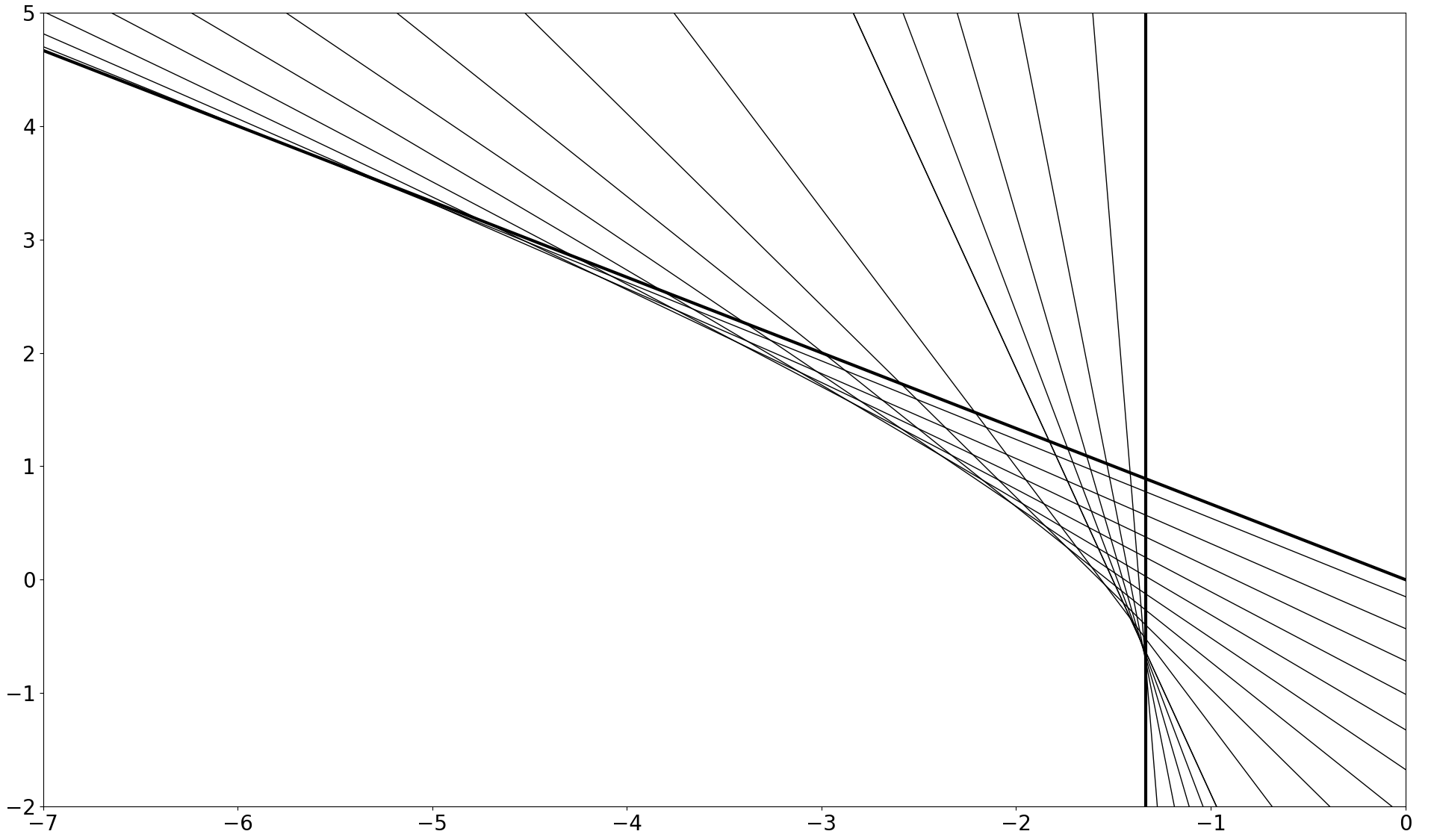}
        \caption{Illustration of the family $\cF = (\ell_x)_{x \in [0,\infty]}$ of lines associated the Vasicek model in the forward case. The limiting lines $\ell_0$ and $\ell_\infty$ are indicated in bold and the envelope $\eta$ of $\cF$ appears as the boundary of the region filled by $\cF$.}
        \label{fig:Envelope}
    \end{figure}
 Importantly, a non-linear contour appears at the boundary of the region covered by $\cF$ in \cref{fig:Envelope}. This contour, usually called the \emph{envelope} of $\cF$, is a classic object of interest in the geometry of lines and curves in the plane, see~\cite[Ch.~5]{bruce1992curves}. It can also be seen that this contour, together with the lines $\ell_0$ and $\ell_\infty$ partitions the plane in regions, which are crossed once, twice, or never by the family $\cF$. As discussed above, this corresponds to state-contingent term structure shapes with a single, two, or no local extremum. Extending this basic idea, it turns out that the envelope is indeed the key to a complete description of state-contingent term-structure shapes. Our main results, Theorem~\ref{thm:main} and \ref{thm:split}, can be summarized as follows:
    \begin{itemize}
 \item The shape of the forward (or yield) curve $x \mapsto f(x,Z_t)$, conditional on $Z_t = z$, is uniquely determined by the position of $z$ relative to $\ell_0$ and $\ell_\infty$, and by the winding number of the `augmented envelope' $\hat \eta$ (the envelope closed by appending segments of $\ell_0$ and $\ell_\infty$) around $z$.
\item The envelope $\eta$, together with the limiting lines $\ell_0$ and $\ell_\infty$, partitions the state space $Z = \RR^2$ into regions $(R_1, \dotsc, R_K)$, each of which corresponds to a certain shape of the forward (or yield) curve $x \mapsto f(x,Z_t)$ appearing conditional on $Z_t \in R_i$. 
 \end{itemize}
 
 We remark that Figure~\ref{fig:Envelope} shows the simplest situation and that the envelope can display self-intersections and other singular points in other cases; see Figure~\ref{fig_ss} and \ref{fig:sp}.

\section{The envelope and related concepts}\label{sec:def}
Although we are eventually interested in the Vasicek model, we will use a more general setup in this section, which may prove useful to analyze other affine term structure models.

\subsection{A family of lines}

    Let $a,b,c:[0,\infty)\rightarrow\RR$ be continuous and on $(0,\infty)$ twice differentiable functions which satisfy 
    \begin{enumerate}[\bf ({A}1)]

        \item $b(x) \neq 0$ and $c(x) \neq 0$ for all $x \in (0,\infty)$.\label{item:nonzero}
    \end{enumerate}
    We consider the family $\cF = (\ell_x)_{x\in (0,\infty)}$ of lines where the line $\ell_x$ consists of all points $z = (z_1, z_2)$ satisfying
    \begin{equation}\label{eq:family}
        F(x,z) := a(x) + b(x)z_1 + c(x) z_2 = 0.
    \end{equation}
	Each line $\ell_x$ also defines a half-space $\ell_x^+$ where
    \begin{align*}
        F(x,z) = a(x) + b(x)z_1 + c(x) z_2 > 0,
    \end{align*}
    and a half-space $\ell_x^-$ with the opposite inequality. 

We would like to have well-defined lines $\ell_0$ and $\ell_\infty$ in the limiting cases, as $x$ tends to the boundaries of $(0,\infty)$. To this end, we introduce the following assumption:
    \begin{enumerate}[resume*]
     \item For $a, b, c$ and their first derivatives the limits $x\rightarrow\delta$, $\delta \in \set{0,\infty}$ exist and are finite. Moreover, 
        $\lim_{x \to \delta}\left(b(x),c(x)\right) \neq 0 $ for $\delta \in \set{0,\infty}$.\label{item:limits}
        \end{enumerate}
        
\begin{rem}\label{rem:scaling}
Using a positive scaling function $\gamma: (0,\infty) \to (0,\infty)$ we can set 
\[\left(\hat a(x), \hat b(x), \hat c(x)\right) = \gamma(x) \left(a(x), b(x), c(x)\right).\]
Note that the corresponding $\hat F(x,z) = \gamma(x) F(x,z)$ defines the same family of lines (and the same half spaces) as $F(x,z)$. Thus, even if condition \ref{item:limits} is not satisfied for the original functions $(a,b,c)$, it may be satisfied after suitable rescaling. In concrete models, after verifying \ref{item:limits} by scaling, it can be convenient to continue working with the unscaled functions. 
\end{rem}

 Under \ref{item:limits}, we denote by $\ell_0$ and $\ell_\infty$ the lines described by the equations
\[F(\delta,z) = a(\delta) + z_1 b(\delta) + z_2 c(\delta) = 0, \qquad \delta \in \set{0,\infty},\]
and by $\ell_0^\pm$ and $\ell_\infty^\pm$ the corresponding half-spaces. We will need the following non-degeneracy condition:
\begin{enumerate}[resume*]
\item The lines $\ell_0$ and $\ell_\infty$ are not parallel,\label{item:intersection}
\end{enumerate}
under which the lines $\ell_0$ and $\ell_\infty$ have a unique intersection point $M$.

   \subsection{Envelope and augmented envelope}
   We now consider the envelope of the family $\cF$ and refer to \cite{bruce1981envelope} and \cite{bruce1992curves} for background and general theory on families of curves in the plane. 
   \begin{defn}\label{def:envelope}
   The \textbf{envelope} $\eta$ of the family $\cF$ consists of all points $z = (z_1,z_2)$, which satisfy \eqref{eq:family} and, in addition, 
    \begin{equation}\label{eq:envelope}
        \partial_x F(x,z) = a'(x) + b'(x)z_1 + c'(x) z_2 = 0.
    \end{equation}
    \end{defn}
    The boundary of the region covered by $\cF$, as indicated in Figure~\ref{fig:Envelope}, will always be part of the envelope, cf.~\cite[Ch.~5]{bruce1992curves}. Together, equations~\ref{eq:family} and \ref{eq:envelope} can be written as a system
    \begin{equation}\label{eq:system}
    \begin{pmatrix}b(x) & c(x) \\ b'(x) & c'(x) \end{pmatrix} \begin{pmatrix}z_1\\z_2\end{pmatrix} = - \begin{pmatrix}a(x)\\a'(x)\end{pmatrix}.
    \end{equation}
Introducing the notation
\[W(f_1, f_2, \dotsc, f_n)(x) = \det \begin{psmallmatrix} f_1(x) & f_2(x) & \dotsm & f_n(x)\\f'_1(x) & f'_2(x) & \dotsm & f'_n(x) \\ ... \\ f^{(n)}_1(x) & f^{(n)}_2(x) & \dotsm & f^{(n)}_n(x) \end{psmallmatrix}\]
for the \emph{Wronskian determinant} of a tuple of functions $f_1\dotsc, f_n$, the following assumption guarantees uniqueness for the solution of \eqref{eq:system}:            \begin{enumerate}[resume*]
        \item $W(b,c)(x) \neq 0$ for all $x > 0$.\label{item:Wbc}
    \end{enumerate}
    \begin{lem}\label{lem:envelope_slope} Under assumptions \ref{item:nonzero} and \ref{item:Wbc} the envelope $\eta = (\eta(x))_{x \in (0,\infty)}$ of $\cF$ is a continuously differentiable curve in $\RR^2$ given by 
        \begin{align*}
        \eta(x) = \left(-\frac{W(c,a)(x)}{W(b,c)(x)}, -\frac{W(a,b)(x)}{W(b,c)(x)}\right)
    \end{align*}
    and with tangent vector 
        \begin{align*}
        \eta'(x) = \frac{W(a,b,c)(x)}{W(b,c)(x)^2} \begin{pmatrix}-c(x)\\b(x)\end{pmatrix}.
    \end{align*}
    Moreover, the slope function $x \mapsto s(x) := -\frac{b(x)}{c(x)}$ of the lines $(\ell_x)_{x \in (0,\infty)}$ is monotone.
    \end{lem}
    \begin{rem}
    Note that $\eta'(x)$ is parallel to $\ell_x$, i.e. the envelope meets $\ell_x$ \emph{tangentially} in $\eta(x)$. This is a general property of envelopes, cf.~\cite[Ch.~5]{bruce1992curves}.
    \end{rem}
    \begin{proof}
    The equation for $\eta$ follows by applying Cramer's rule to the system~\eqref{eq:system}. The equation for $\eta'$ follows by taking derivatives and applying the identities
    \begin{align*}
    c W(a,b,c) &= \frac{\mathrm{d}W(c, a)}{\mathrm{d}x}W(b,c) - W(c,a) \frac{\mathrm{d}W(b,c)}{\mathrm{d}x}\\
    -b W(a,b,c) &= \frac{\mathrm{d}W(a,b)}{\mathrm{d}x} W(b,c) - W(a,b) \frac{\mathrm{d}W(b,c)}{\mathrm{d}x},
    \end{align*}
    which follow from the Desnanot-Jacobi determinant identity (see \cite{brualdi1983determinantal}), but can also be verified by direct calculation. Finally, by \ref{item:nonzero}, the slope function is always well-defined. Taking derivatives, we obtain
    \[s'(x) = \frac{W(b,c)(x)}{c(x)^2}.\]
    By \ref{item:Wbc} the numerator has constant sign; therefore $s$ must be monotone.
    \end{proof}
The tangent vector $\eta'(x)$ of the envelope vanishes only when $W(a,b,c)(x) = 0$, or equivalently, when
    \begin{equation}\label{eq:regular}
        \partial_{xx} F(x,z) = a''(x) + b''(x)z_1 + c''(x) z_2 =  0.
    \end{equation}
    Points on $\eta$ where \eqref{eq:regular} (and, by definition, also \eqref{eq:family} and \eqref{eq:envelope}) holds, are called \emph{points of regression}. All other points of $\eta$ are called \emph{regular}. We assume that there are only finitely many points of regression, which is guaranteed by the following assumption:
    \begin{enumerate}[resume*]
        \item $W(a,b,c)(x) = 0$ for only finitely many $x > 0$.\label{item:Wabc}
\newcounter{saveenum}
  \setcounter{saveenum}{\value{enumi}}
    \end{enumerate}

Finally, we are interested in the limiting behaviour of $\eta(x)$ as $x$ tends to zero or infinity, i.e., how it `connects' to the lines $\ell_0$ and $\ell_\infty$.
\begin{lem}\label{lem:contact}Let $\delta \in \set{0,\infty}$ and assume \ref{item:nonzero}, \ref{item:limits} and \ref{item:Wbc}. Either the contact point $\eta(\delta) := \lim_{x \to \delta}\eta(x)$ exists in $\RR^2$ and is located on $\ell_\delta$ or 
$\lim_{\delta \to \infty} |\eta(x)| = \infty$ and $\eta(x)$ approaches $\ell_\delta$ asymptotically as $x \to \delta$, i.e. 
\[\lim_{x \to \delta} d(\eta(x),\ell_\delta) = 0,\]
with $d(.,.)$ denoting distance.
\end{lem}
\begin{rem}\label{rem:asymp}
Even if $\lim_{\delta \to \infty} |\eta(x)| = \infty$, we will denote by $\eta(0)$ (or $\eta(\infty)$) the `asymptotic contact point' of $\eta$ and $\ell_0$ (or $\ell_\infty$). The correct interpretation will usually be clear from the context. An example of asymptotic contact at $\eta(\infty)$ is shown in Figure~\ref{fig_ss}.
\end{rem}
\begin{proof}If $\eta(\delta) := \lim_{x \to \delta}\eta(x)$ exists in $\RR^2$, then by continuity
\[F(\delta,\eta(\delta)) = \lim_{x \to \delta} F(x,\eta(x)) = 0,\]
showing that $\eta(\delta) \in \ell_\delta$. Even if $\lim_{x \to \delta}|\eta(x)| = \infty$, continuity of $(a,b,c)$ and \ref{item:limits} allows us to conclude that
\[\lim_{x \to \delta} d(\eta(x),\ell_\delta) = \lim_{x \to \delta} \frac{|F(\delta,\eta(x))|}{\sqrt{b(\delta)^2 + c(\delta)^2}} = \frac{\lim_{x \to \delta} |F(x,\eta(x))|}{\sqrt{b(\delta)^2 + c(\delta)^2}}  = 0. \qedhere\]
\end{proof}

In general, $\eta(0) \neq \eta(\infty)$, and the envelope is not a closed curve. However, we can close it by augmenting it with segments of the lines $\ell_0$ and $\ell_\infty$. 
\begin{defn}\label{defn:augmented}Let assumptions \ref{item:nonzero} - \ref{item:Wabc} hold, and let $M$ be the intersection point of $\ell_0$ and $\ell_\infty$. The \textbf{augmented envelope} $\hat \eta $ is the oriented curve with basepoint $M$, piecewise defined by 
\begin{enumerate}[(a)]
\item The line segment from $M$ to $\eta(0)$ (contained in $\ell_0$);
\item The envelope $\eta$; 
\item The line segment from $\eta(\infty)$ to  $M$ (contained in $\ell_\infty$).  
\end{enumerate}
\end{defn}
If necessary, we can parameterize $\hat \eta$ by a parameter $t$ in $I = [-1,0) \cup (0,1) \cup (1,2]$, such that the three parts of $I$ correspond to the pieces (a), (b) and (c) above. If $\eta(0)$ and $\eta(\infty)$ are proper points, we can replace $I$ by $\bar I = [-1,2]$ and $\hat \eta$ is a piecewise regular, continuous closed curve in the usual sense. If $\eta$ is only asymptotically approaching $\ell_0$ or $\ell_\infty$, then the line segments from (a) or (c) can be unbounded and $\hat \eta$ is not continuous (or closed) in the usual sense. However, for our purposes, no difference needs to be made between these cases. In particular the \emph{winding number} of $\hat \eta$ around any point $z \in \RR^2$ is always well-defined. Intuitively, the winding number is the number turns of $\hat \eta$ around $z$, with counterclockwise turns contributing positively, and clockwise turns negatively. The following rigorous definition is adapted from \cite[Ch.5-7]{do2016differential}:
 
    \begin{defn}\label{rem:wind}
    \begin{enumerate}[(a)]
    \item Let $z \in \RR^2$ and $z \not \in \hat \eta$. The \emph{position map} of $\hat \eta$ relative to $z$ is the continuous mapping $\phi$ from $\hat I = [-1,2]$ to the unit circle $S_1$ given by 
    \begin{align*}
            \phi(t) = \frac{\hat \eta(t) - z}{\left|\hat \eta(t)- z\right|}.
        \end{align*}
    \item The \textbf{winding number} $\wind_{\hat \eta}(z)$ of $\hat \eta$ around $z$ is the topological degree of the position map $\phi$.
     \end{enumerate}
     \end{defn}
In order for the topological degree to be well-defined we need $\phi$ to be closed ($\phi(-1) = \phi(2)$)  and continuous. The first property follows directly from Definition~\ref{defn:augmented}. For the second property we emphasize that the position map $\phi$ is \emph{always} continuous due to Lemma~\ref{lem:contact}, even when $\hat \eta$ contains asymptotic points. Therefore, the winding number is always well-defined. In particular, there is a unique (mod $2 \pi$) continuous representation $\theta$ of $\phi$ in polar coordinates,
\[\phi(t) = \left(\cos(\theta(t)), \sin(\theta(t)\right)),\quad t \in I = [-1,2],\]
and $\wind_{\hat \eta}(z) = \tfrac{1}{2\pi}(\theta(2) - \theta(-1))$.

    \begin{figure}
        \centering
        \includegraphics[width=\textwidth]{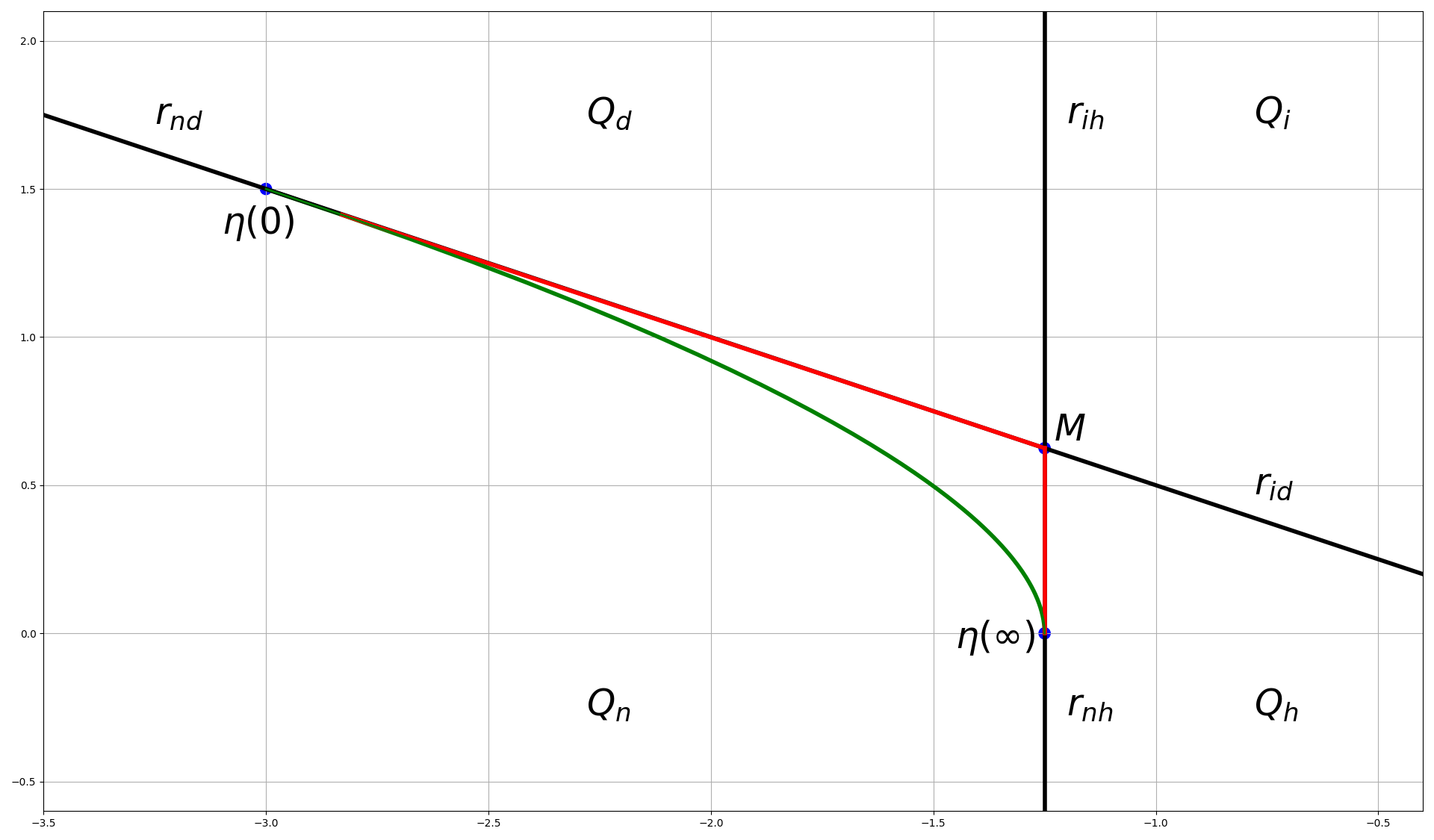}
        \caption{The envelope (green), augmented envelope (green + red) and augmented envelope with whiskers (green + red + black). Also indicated are the envelope's contact points $\eta(0)$ and $\eta(\infty)$, as well as our notation for the quadrants $Q$ and their boundaries $r$.}
        \label{fig:aeww}
    \end{figure}
    
    \subsection{Helper lines and tracking functions}\label{sec:helper}
    Some key arguments in the proof of our main results in Section~\ref{sec:main} will be based on observing a helper line $h$ not belonging to the family $\cF = (\ell_x)_{x \in [0,\infty]}$. We introduce some related concepts:
    
     \begin{defn}
        A line $h$ is in \textbf{oblique position} with respect to the family $\cF = (\ell_x)_{x \in [0,\infty]}$, if it is not parallel to $\ell_x$ for any $x \in [0,\infty]$.
    \end{defn}
 In order to find an oblique helper line, we require the following:
 
     \begin{enumerate}[resume*]
        \item There exists a line $h$ in oblique position with respect to $\cF$. \label{item:oblique}
    \end{enumerate}
    
    \begin{defn}
        A line $h$ is in \textbf{general position} with respect to the augmented envelope $\hat \eta$, if it meets $\hat \eta$ only in regular points, never tangentially, and not in the points $\eta(0)$, $\eta(\infty)$ or $M$.
    \end{defn}
    The above definition makes sure that in an intersection point of $h$ with $\hat{\eta}$, the line always cuts transversally through the augmented envelope. 

    \begin{defn}\label{def:leftright}
        We consider a \emph{reference point} $z \not \in \hat \eta$ and an (oriented) \emph{helper line} $h$ through $z$ with direction vector $v$, given in parametric form as
        \[h = \set{z(\tau) = z + \tau v: \quad \tau \in \RR}.\] 
        We say that a point on $h$ is `to the left of $z$' if it can be represented with $\tau < 0$ and `to the right of $z$' if $\tau > 0$.
        We assume that $h$ is in oblique position with respect to $\cF$. Under this assumption the intersection $\ell_x \cap h$ consists of a single point
        \[h \cap \ell_x = \set{z + \tau(x) v},\] 
        for any $x \in \RR$ and defines the \textbf{tracking function} $x \mapsto \tau(x)$, which `tracks' the evolution of $(\ell_x)_{x \in [0,\infty]}$ along $h$. 
    \end{defn}
    The tracking function has the explicit form
    \begin{align*}
        \tau(x) = \frac{a(x) + b(x)z_1 + c(x)z_2}{b(x)v_1 +c(x)v_2}
    \end{align*}
    and is differentiable since $a, b$ and $c$ are differentiable and $h$ is oblique, i. e. the denominator has no zeros. We make the following observation:
    \begin{lem}\label{lem:tracking}
    Let \ref{item:nonzero} - \ref{item:oblique} hold, and let $h$ be a helper line through $z \in \RR^2$ in oblique position with associated tracking function $\tau$. Then, for any $x \in (0,\infty)$,
    \begin{align}
        \tau'(x) = 0 &\qquad \Longleftrightarrow \qquad z + \tau(x)v = \eta(x).\\
        \intertext{In addition, if $\tau'(x) = 0$, then}
        \tau''(x) = 0 &\qquad \Longleftrightarrow \qquad z + \tau(x)v \text{ is a regression point of $\eta$.}
        \end{align}
    \end{lem}
In other words, the critical points of the tracking functions are precisely the intersection points of the envelope $\eta$ and $h$, and a critical point is extremal if and only if it is a regular point of $\eta$. 
\begin{proof}
By definition, the tracking function $\tau$ satisfies
\[a(x)  + b(x) (z_1 + \tau(x) v_1) + c(x)(z_2 + \tau(x) v_2) = 0\]
for all $x \in (0,\infty)$. Taking derivatives, we obtain
\begin{multline}\label{eq:diff}
a'(x)  + b'(x) \Big(z_1 + \tau(x) v_1\Big) + c'(x)\Big(z_2 + \tau(x) v_2\Big)  + \\ + \tau'(x)\Big(b(x)v_1 + c(x)v_2\Big) = 0.
\end{multline}
If $\tau'(x) = 0$, the last term vanishes, showing that $z + \tau(x) v$ is part of the envelope, i.e., that $z + \tau(x) v = \eta(x)$. On the other hand, $z + \tau(x) v = \eta(x)$ implies that $\tau'(x)(b(x)v_1 + c(x)v_2) = 0$. By the obliqueness condition, the second term can never vanish, implying that $\tau'(x) = 0$; showing the first claim. For the second claim, take another derivative in \eqref{eq:diff} to obtain
\begin{multline}\label{eq:diff_twice}
a''(x)  + b''(x) \Big(z_1 + \tau(x) v_1\Big) + c''(x)\Big(z_2 + \tau(x) v_2\Big)  + \\ + 2\tau'(x)\Big(b'(x)v_1 + c'(x)v_2\Big)  + \tau''(x)\Big(b(x)v_1 + c(x)v_2\Big) = 0.
\end{multline}
Since $\tau'(x) = 0$, the second-to-last term vanishes. If also $\tau''(x) = 0$, then also the last term vanishes, and -- comparing with \eqref{eq:regular} -- we conclude that $z + \tau(x) v = \eta(x)$ is a regression point. Conversely, if $z + \tau(x) v$ is a regression point of the envelope, then $\tau''(x)(b(x)v_1 + c(x)v_2) = 0$. The obliqueness condition now implies that $\tau''(x) = 0$, as claimed.
\end{proof}

\subsection{Verifying the assumptions for the Vasicek model}
We verify that all of the above assumptions are satisfied in the two-factor Vasicek model. We set
\[m_{\rm f}(x,\lambda)  = -\lambda e^{-\lambda x} \quad \text{and} \quad m_{\rm y}(x,\lambda) = \frac{e^{-\lambda x}(1 + \lambda x)  - 1}{\lambda x^2}.\]
From Section~\ref{sec:bg} we obtain that 
\[b_{\rm f}(x) = m_{\rm f}(x,\lambda_1) = -\lambda_1 e^{-\lambda_1 x}, \qquad c_{\rm f}(x) = m_{\rm f}(x,\lambda_2) = -\lambda_2 e^{-\lambda_2 x}\]
in the forward case, and 
\[b_{\rm y}(x) = m_{\rm y}(x,\lambda_1) = \frac{e^{-\lambda_1 x}(1 + \lambda_1 x)  - 1}{\lambda_1 x^2}, \quad c_{\rm y}(x) = m_{\rm y}(x,\lambda_2) = \frac{e^{-\lambda_2 x}(1 + \lambda_2 x)  - 1}{\lambda_2 x^2}\]
in the yield case. Defining the operator
\begin{multline}\label{eq:A_operator}
\left(\cA m \right)(x) := -\Big\{ \frac{\sigma_1^2}{2\lambda_1^2} m(x,2 \lambda_1)  + \frac{\sigma_2^2}{2\lambda_2^2} m(x,2 \lambda_2) + \\ + r m(x,\lambda_1 + \lambda_2) + u_1 m(x,\lambda_1) + u_2 m(x,\lambda_2)\Big\},
\end{multline}
we can write
\[a_{\rm f}(x) = \left(\cA m_{\rm f}\right)(x) \qquad \text{and} \qquad a_{\rm y}(x) = \left(\cA m_{\rm y}\right)(x).\]
Clearly \ref{item:nonzero} is satisfied in both cases of yield and forward curve. It is also easy to see that \ref{item:limits} is satisfied for $(b_{\rm y}, c_{\rm y})$. As for the forward curve case, scaling with $\gamma(x) = e^{\lambda_1 x}$ and using that $\lambda_1 < \lambda_2$ we see that
$\hat b_{\rm f}(x) = \gamma(x) b_{\rm f}(x) = 1$ and $\hat c_{\rm f}(x) = \gamma(x) c_{\rm f}(x) = -\lambda_2 e^{(\lambda_1 - \lambda_2)x}$ also satisfy \ref{item:limits}. As pointed out in Rem.~\ref{rem:scaling} we can continue working with the unscaled functions.  For $\ell_0$ we obtain the equation
\[\ell_0: \lambda_1 \theta_1 + \lambda_2 \theta_2 - \lambda_1 z_1 - \lambda_2 z_2 = 0\]
for both the yield and forward curve case. For $\ell_\infty$ we obtain
\begin{align*}
\ell_\infty^{\rm f}: 0 &= u_1 - z_1\\
\ell_\infty^{\rm y}: 0 &= \sigma_1^2 + \sigma_2^2 + \tfrac{r}{\lambda_1 + \lambda_2} + \tfrac{u_1}{\lambda_1} + \tfrac{u_1}{\lambda_1} - \tfrac{1}{\lambda_1} z_1 - \tfrac{1}{\lambda_2} z_2.
\end{align*}
In any case, \ref{item:intersection} is satisfied. Next, we need to compute the Wronskians $W(b,c)$ and $W(a,b,c)$.  For the forward curve case we obtain
    \begin{align*} 
        W(b_{\rm f}, c_{\rm f})(x) = \text{e}^{-(\lambda_1 + \lambda_2)x}\lambda_1\lambda_2(\lambda_1 - \lambda_2),
    \end{align*}
    which never vanishes; hence \ref{item:Wbc} holds. 
    In the yield case we have
    \[W(b_{\rm y}, c_{\rm y})(x) = \frac{e^{-x (\lambda_1 + \lambda_2)}}{x^2} \left(\tfrac{\lambda_1}{\lambda_2}  - \tfrac{\lambda_2}{\lambda_1} + (\lambda_1 - \lambda_2) x + \tfrac{\lambda_2}{\lambda_1}\text{e}^{\lambda_1x} - \frac{\lambda_1}{\lambda_2}\text{e}^{\lambda_2x}\right).\]
The term in brackets vanishes for $x = 0$; differentiating it with respect to $x$ gives
    \begin{align*}
        (\lambda_1 - \lambda_2)  + \lambda_2\text{e}^{\lambda_1x} - \lambda_1\text{e}^{\lambda_2x} = \lambda_1\lambda_2\sum_{k=1}^\infty \frac{x^k\left(\lambda_2^{k-1} - \lambda_1^{k-1}\right)}{k!} > 0,\ \forall x > 0.
    \end{align*}
We conclude that the Wronskian determinant does not have zeros for $x > 0$, and hence that \ref{item:Wbc} holds also in the yield curve case. Extending the Wronskian by $a$, we obtain
    \begin{align}
        \label{eq:detfor}
        W(a_{\rm f},b_{\rm f},c_{\rm f})(x) &= -\sigma_2^2(2\lambda_2 - \lambda_1)\text{e}^{-2\lambda_2x} - \rho\sigma_1\sigma_2(\lambda_1 + \lambda_2)\text{e}^{-(\lambda_1 + \lambda_2)x} -\\\notag&\quad - \sigma_1^2(2\lambda_1 - \lambda_2)\text{e}^{-2\lambda_1x}.
    \end{align}  This function has at most two zeros for $x > 0$;  this follows from the fact that the functions $(e^{-2\lambda_2x},e^{-2(\lambda_2  + \lambda_1)x}, e^{-2\lambda_1x})$ form a Descartes system, see \cite[Sec.~3.2]{borwein1995polynomials}, \cite[\S1.4]{karlin1968total} and also \cite{keller2021classification}. Hence \ref{item:Wabc} holds true in the forward curve case. In the yield curve case, $W(a_{\rm y},b_{\rm y},c_{\rm y})$ is much more difficult to compute. However, as shown in the appendix, an argument based on total positivity shows that it also has at most two zeros, and hence that \ref{item:Wabc} is satisfied. Finally we turn to \ref{item:oblique}, i.e., to the existence of an oblique line. Computing the slopes of the lines $\ell_x$, we obtain
    \begin{align*}
        -\frac{b_{\rm f}(x)}{c_{\rm f}(x)} = -\frac{\lambda_1}{\lambda_2}\text{e}^{(\lambda_2-\lambda_1)x}
    \end{align*}
    in the forward case, and 
    \begin{align*}
        -\frac{b_{\rm y}(x)}{c_{\rm y}(x)} = -\frac{\lambda_2\left((1+\lambda_1x)\text{e}^{-\lambda_1x} - 1\right)}{\lambda_1\left((1+\lambda_2x)\text{e}^{-\lambda_2x} - 1\right)},
    \end{align*}
    in the yield curve case. By Lemma~\ref{lem:envelope_slope}, the slopes are strictly monotone; in fact, the negative sign of $W(b,c)$ shows that they are strictly decreasing.
Evaluating the boundary values at $0$ and $\infty$, we conclude that the slope corresponding to the forward curve takes values in $[-\frac{\lambda_1}{\lambda_2}, -\infty]$ and the slope corresponding to the yield curve takes values in $[-\frac{\lambda_1}{\lambda_2}, -\frac{\lambda_2}{\lambda_1}]$. In both cases 0 is not an element of these intervals, so we can choose any horizontal line $z_2 = K \in \RR$ as our oblique line. We sum up our analysis in the following result:
    \begin{lem}\label{lem:vasicek}
    The families $(\ell^{\rm f}_x)_{x \in (0,\infty)}$ and $(\ell^{\rm y}_x)_{x \in (0,\infty)}$, associated to the forward curve resp.{} yield curve in the two-dimensional Vasicek model satisfy all properties \ref{item:nonzero} -- \ref{item:oblique}. Their slope functions $x \mapsto s(x) = -\tfrac{b(x)}{c(x)}$ are strictly decreasing and any horizontal line is oblique.
    \end{lem}

The envelopes can be explicitly calculated in both the forward and the yield case. Using the operator $\cA$ from \eqref{eq:A_operator}, we can express the envelope in the forward case as 
	\begin{align*}
	    \eta_1^{\rm f}(x) = -\frac{\left(\cA n_1^{\rm f}\right)(x)}{\lambda_1(\lambda_1 - \lambda_2)}, \qquad  \eta_2^{\rm f}(x) = -\frac{\left(\cA n_2^{\rm f}\right)(x)}{\lambda_2(\lambda_1 - \lambda_2)}
	\end{align*}
	where
	\begin{align*}
	    n_1^{\rm f}(x,\lambda) = \lambda(\lambda - \lambda_2)\text{e}^{(\lambda_1 - \lambda)x},\quad 
      n_2^{\rm f}(x,\lambda) = \lambda(\lambda - \lambda_1)\text{e}^{(\lambda_2 - \lambda)x}.
	\end{align*}
We omit the expression for the yield case, since it is quite complicated.

\subsection{Classification of singular points}

We give a classification of singular points of the envelope, which will provide some insight into the possible decompositions of the state space. As in \cite{keller2021classification}, we distinguish the following regimes of the model:

\begin{defn}The two-dimensional Vasicek model is called \label{def:regime}
\begin{itemize}
\item \textbf{scale-separated}, if $2 \lambda_1 < \lambda_2$,
\item \textbf{scale-proximal}, if $2 \lambda_1 > \lambda_2$, and
\item \textbf{scale-critical}, if $2 \lambda_1 = \lambda_2$.
\end{itemize}
\end{defn}
Moreover, the correlation parameter $\rho$ will play an important role. 

The following lemma is evident from the form of the tangent vector $\eta'$ as given in \cref{lem:envelope_slope}. Note that a point of a differentiable curve is called singular if the tangent vector vanishes, and such a point is called a cusp, if the tangent vector changes direction after passing through it. 
\begin{lem}\label{lem:singular0}
The following are equivalent for a point $z = \eta(x)$ on the envelope:
\begin{enumerate}[(a)]
\item $W(a,b,c)(x) = 0$;
\item $z$ is a point of regression;
\item $z$ is a singular point of $\eta$;
\end{enumerate}
Moreover, $z$ is a cusp point iff $x$ is a transversal zero of $W(a,b,c)$.
\end{lem}

\begin{lem}\label{lem:singular}In the two-dimensional Vasicek model, consider the functions $(a,b,c)$, related to either the forward curve or to the yield curve. Let $N$ be the number of zeroes of $(0,\infty) \ni x \mapsto W(a,b,c)(x)$. 
\begin{enumerate} 
\item In the scale-proximal case with $\rho \ge 0$, it holds that $N = 0$. The envelope has no self-intersections and no cusps. It also has no intersections with $\ell_0$ or $\ell_\infty$.
\item In the scale-separated case, it holds that $N \le 1$. The envelope has no self-intersections and it has $N$ cusps. It has at most $N$ intersections with each of $\ell_0$ and $\ell_\infty$. 
\item In the scale-proximal case with $\rho < 0$ , it holds that $N \le 2$. The augmented envelope has at most $\lfloor N/2 \rfloor \le 1$ self-intersections and at most $N$ cusps. In addition it has at most $N$ intersections with each of $\ell_0$ and $\ell_\infty$
\end{enumerate}
The scale-critical case behaves like the scale-proximal case if $\rho \ge 0$ and like the scale-separated case if $\rho < 0$.
\end{lem}
\begin{proof} First consider the Wronskian $W(a,b,c)(x)$ for the forward curve, as given in \cref{eq:detfor}. The functions $(e^{-2\lambda_2x},e^{-2(\lambda_2  + \lambda_1)x}, e^{-2\lambda_1x})$ form a Descartes system, see \cite[Sec.~3.2]{borwein1995polynomials}. By the variation-diminishing property of Descartes systems, see \cite[Thm.~3.2.4]{borwein1995polynomials}, the number $N$ of zeroes of $x \mapsto W(a,b,c)(x)$ is bounded by the number of sign changes in the sequence of coefficients of the Descartes system. Thus, the bounds for $N$ in the three cases of \cref{def:regime} follow immediately from \cref{eq:detfor}. The same bounds must holds for the yield curve, by \cref{lem:start}.\\
Next, we focus on self-intersections: For fixed $\xi \in (0,\infty)$, consider the line
\[\ell_\xi: a(\xi) + b(\xi)z_1 + c(\xi) z_2 = 0.\]
Clearly, $\ell_\xi$ meets $\eta$ at $\eta(\xi)$ and is tangent there. To evaluate the position of the full envelope $\eta$ relative to $\ell_\xi$ we set
\[f_\xi(x) = a(\xi) + b(\xi)\eta_1(x) + c(\xi) \eta_2(x),\]
which has a zero at some $x_0$ if and only if $\eta$ meets $\ell_\xi$ at time $x = x_0$. Note that all intersections at $x_0 \neq \xi$ are true crossings due to the strictly decreasing slope of $\eta$; cf. \cref{lem:envelope_slope}. Taking derivatives, we find
\begin{equation}\label{eq:fxi_diff}
f'_\xi(x) = \frac{W(a,b,c)(x)}{W(b,c)^2} m_\xi(x), \quad \text{where } m_\xi(x) = c(\xi)b(x) - b(\xi)c(x).
\end{equation}
Note that the only zero of $m_\xi(x)$ occurs at $x = \xi$, other zeroes are ruled out by the decreasing-slope property from Lemma~\ref{lem:envelope_slope}. Thus any zero of $f'_\xi$ at $x \neq \xi$ is also a zero of $W(a,b,c)$. 
Let now $z_*$ be point of self-intersection of $\eta$, i.e., there exist $0 < \xi_1 < \xi_2$ such that $\eta(\xi_1) = \eta(\xi_2) = z_*$. Then, also the tangent line $\ell_{\xi_1}$ has to be crossed again in $\eta(\xi_2)$, i.e. $f_{\xi_1}(\xi_1) = f_{\xi_1}(\xi_2) = 0$. By the mean value theorem there exists $\xi_* \in (\xi_1, \xi_2)$ such that $f'_{\xi_1}(\xi_*) = 0$ (the zero being transversal), which by \eqref{eq:fxi_diff} implies that $\eta$ has a cusp at time $\xi_*$. This shows that any two times of self-intersections must be separated by a cusp; in particular $\eta$ cannot have self-intersections in case (a). For case (b), assume that $\eta$ has a cusp at time $\xi_*$ (i.e., $W(a,b,c)(\xi_*) = 0$) and consider the tangent line $\ell_{\xi_*}$. From \eqref{eq:fxi_diff} we conclude that $f'_{\xi_*}({\xi_*}) = 0$ is the only zero of $f'_{\xi_*}$, since there are no other cusp points. Hence $\ell_{\xi_*}$ is crossed only once by $\eta$; more precisely the part before the cusp and the part after the cusp stay on opposite sides of $\ell_{\xi_*}$. This rules out self-intersections in case (b). Finally, consider case (c), where self-intersections are possible, iff $N=2$. Assume that there are at least two self-intersections points $z$ and $\zeta$. Putting a line $h$ through $z$ and $\zeta$ this line is intersected twice by $\eta$ before its first cusp, and twice after its last cusp. By Cauchy's mean value theorem, we conclude that the tangent vector of $\eta$ must be parallel to the line $h$ twice; once before the first cusp and once after the last cusp. This contradicts the decreasing-slope property and bounds the number of self-intersections by $1 = \lfloor 2/2 \rfloor$.\\
Applying similar arguments as above to the lines $\ell_0$ and $\ell_\infty$ (and using the corresponding functions $f_0(x)$ and $f_\infty(x)$) the number of intersections with $\ell_0$ and $\ell_\infty$ can be bounded analogously.
\end{proof}

\section{Main Results}\label{sec:main}

In this section we show the two main results on the state-conditional shape of the forward- and the yield-curve in the two-dimensional Vasicek model. To formulate the results we need to distinguish four subsets (\emph{`quadrants'}) of the state space $Z = \RR^2$; see \cref{fig:aeww} for an illustration.
    \begin{align*}
        Q_{n} &:= {\ell_0^+}\cap {\ell_\infty^+} \quad \text{(normal-like)}\\
        Q_{h} &:= \ell_0^+\cap\ell_\infty^- \quad \text{(hump-like)}\\
        Q_{i} &:= {\ell_0^-}\cap {\ell_\infty^-} \quad \text{(inverse-like)}\\
        Q_{d} &:= \ell_0^-\cap\ell_\infty^+ \quad \text{(dip-like)}
    \end{align*}
    If needed we add superscripts ($Q_n^{\rm f}$,$Q_n^{\rm y}$, \dots), to clarify whether we refer to the family of lines associated to the forward- or to the yield curve. The four quadrants are determined by the initial sign and the terminal sign of the forward (or yield) curve's derivative. In $Q_n^{\rm f}$, for example, the forward curve must be initially increasing (for $x$ in a right neighborhood of $0$) and also terminally increasing as $x \to \infty$. This behavior is the same as for a normal curve, hence we call it the `normal-like' quadrant, etc.  The boundaries between the quadrants are given by the rays from $M$ along $\ell_0$ and $\ell_\infty$, for which we introduce the following notation:
  \begin{align*}
r_{nh} &= \partial Q_n \cap \partial Q_h, \qquad &&r_{hi} = \partial Q_h \cap \partial Q_i,\\
r_{id} &= \partial Q_i \cap \partial Q_d, \qquad &&r_{dn} = \partial Q_d \cap \partial Q_n.
\end{align*}
    
Clearly, the unique common point of any two rays is the intersection point $M$ of $\ell_0$ and $\ell_\infty$. Our first main result links the shape of the forward- and the yield-curve to the winding number of the augmented envelope $\hat \eta$ (see \cref{defn:augmented}) and to the quadrants defined above:

        \begin{thm}\label{thm:main}Let $E(z)$ be the number of local extrema of the forward curve $x \mapsto f(x,Z_t)$ or the yield curve $x \mapsto y(x,Z_t)$ in the two-dimensional Vasicek model, conditional on $Z_t = z \in \RR^2$. Let $\hat \eta$ be the augmented envelope of the family $(\ell_x)_{x \in (0,\infty)}$ associated to $f$ or $y$, as defined in \cref{defn:augmented}. Then, for any $z \in \RR^2 \setminus (\eta \cup \ell_0 \cup \ell_\infty)$
        \begin{equation}\label{eq:main_formula}
            E(z) = 2 \abs{\wind_{\hat{\eta}}(z)} + \Ind{Q_{\texttt{h}}\cup Q_{\texttt{d}}}(z),
        \end{equation}
        where $\wind_{\hat{\eta}}(z)$ denotes the winding number of $\hat \eta$ around $z$.
    \end{thm}

This result allows to completely determine the shape of the forward- and the yield-curve conditional on the current state $Z_t \in \RR^2$ of the stochastic model, extending and strengthening the results of \cite{keller2021classification}:
    
    \begin{cor}The shape, i.e., the number and the sequence of local minima and maxima, of $x \mapsto f(x,Z_t)$ and $x \mapsto y(x,Z_t)$, conditional on $Z_t = z \in \RR^2 \setminus (\eta \cup \ell_0 \cup \ell_\infty)$ is uniquely determined by $E(z)$ and by the location of $z$ within the quadrants $Q_n$, $Q_h$, $Q_i$ and $Q_d$.
    \end{cor}

  Our second main result shows that the augmented envelope partitions the state space into finitely many regions of `constant shape':
    
    \begin{thm}\label{thm:split}The augmented envelope $\hat \eta$ partitions the state space $Z = \RR^2$ into finitely many open regions $(R_1, \dotsc, R_K)$, for which the following holds: 
    \begin{enumerate}
\item Within each region the shape, i.e., the number and the sequence of local minima and maxima, of the forward (or yield) curve, conditional on $Z_t = z \in R_i$ stays the same.
\item If two regions $R_i \neq R_j$ share a boundary containing a one-dimensional submanifold of $\RR^2$, then the shapes of the forward (or yield) curve on $R_i$ and $R_j$ must be different.
\end{enumerate}
    \end{thm}
    
In particular, this result allows to quantify the probability of each shape to appear, and also to enumerate the possible transitions between different shapes, see \cref{sec:shape} below. Some examples of state-space decompositions are illustrated in Figures \labelcref{fig:sp0}, \labelcref{fig_ss} and \labelcref{fig:sp}.
    
            \begin{rem}\label{rem:online}A question left open by the results above,  is to determine the shape of the forward/yield curve, conditional on a state $z$ located on the boundary of two regions, or equivalently, for $z \in \eta \cup \ell_0 \cup \ell_\infty$. From the proofs given below it will become obvious that such states correspond to curves with either: critical points, which are not extremal; or a derivative that vanishes as $x \to 0$ or $x \to \infty$. In any case, such behavior does not add a local extremum to the curve, and we arrive at the following rule: The shape of a point $z$ at the boundary between two (or more) regions, is the shape of the adjacent region \emph{with the least number of local extrema}.
\end{rem}

                 \begin{figure}
            \centering
            \includegraphics[width=1\textwidth]{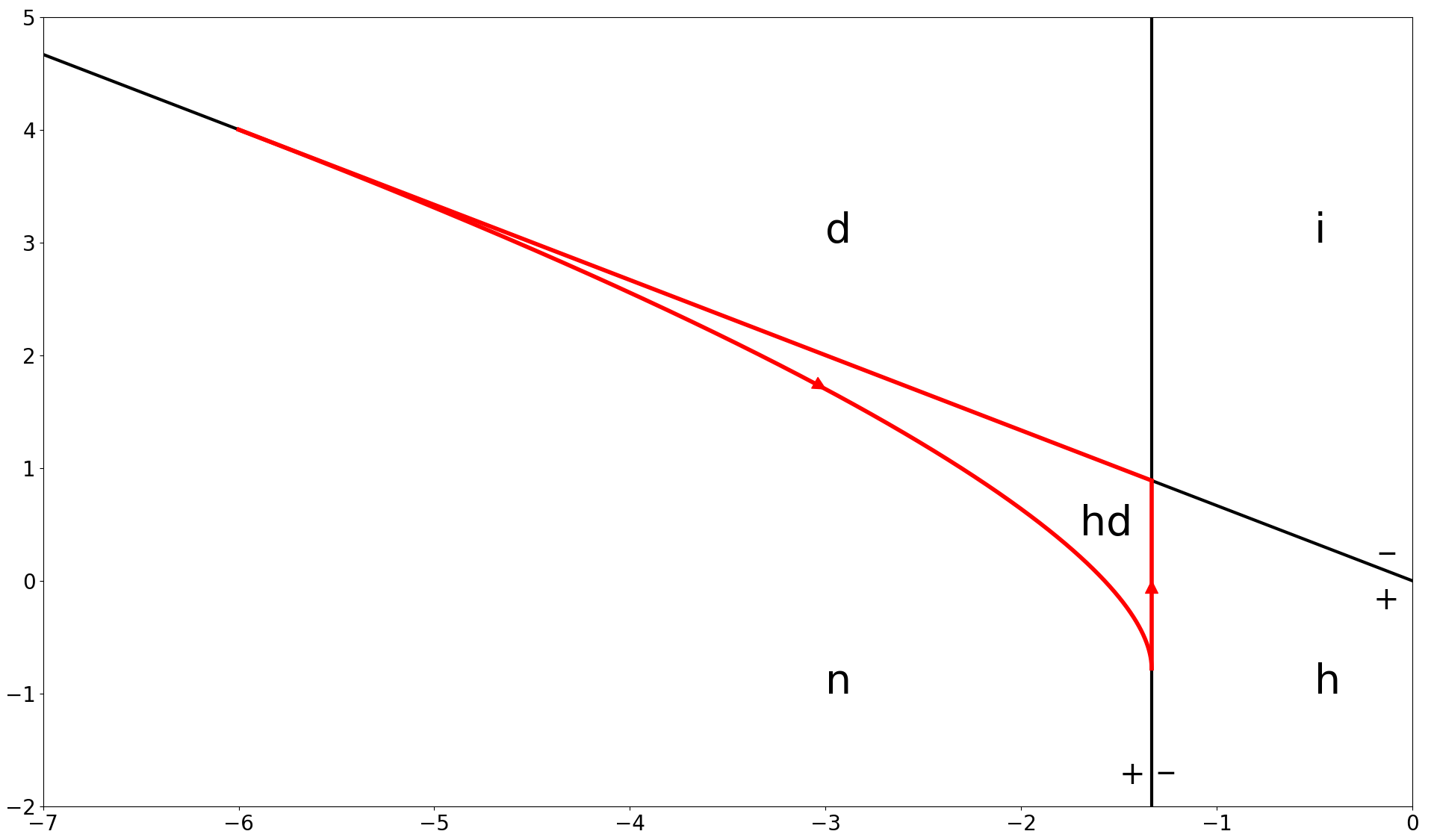}
            \caption{\label{fig:sp0} Example of augmented envelope (red) and state-space decomposition for the forward curve in the scale-proximal case with $\rho \ge 0$. Parameters used are: $\lambda_1 = 1.0$, $\lambda_2 = 1.5$, $\sigma_1 = 1.0$, $\sigma_2 = 1.0$, $\rho=0.5$, $\theta_1=0.0$, $\theta_2=0.0$.}

        \end{figure}
        \begin{figure}
            \centering
            \includegraphics[width=1\textwidth]{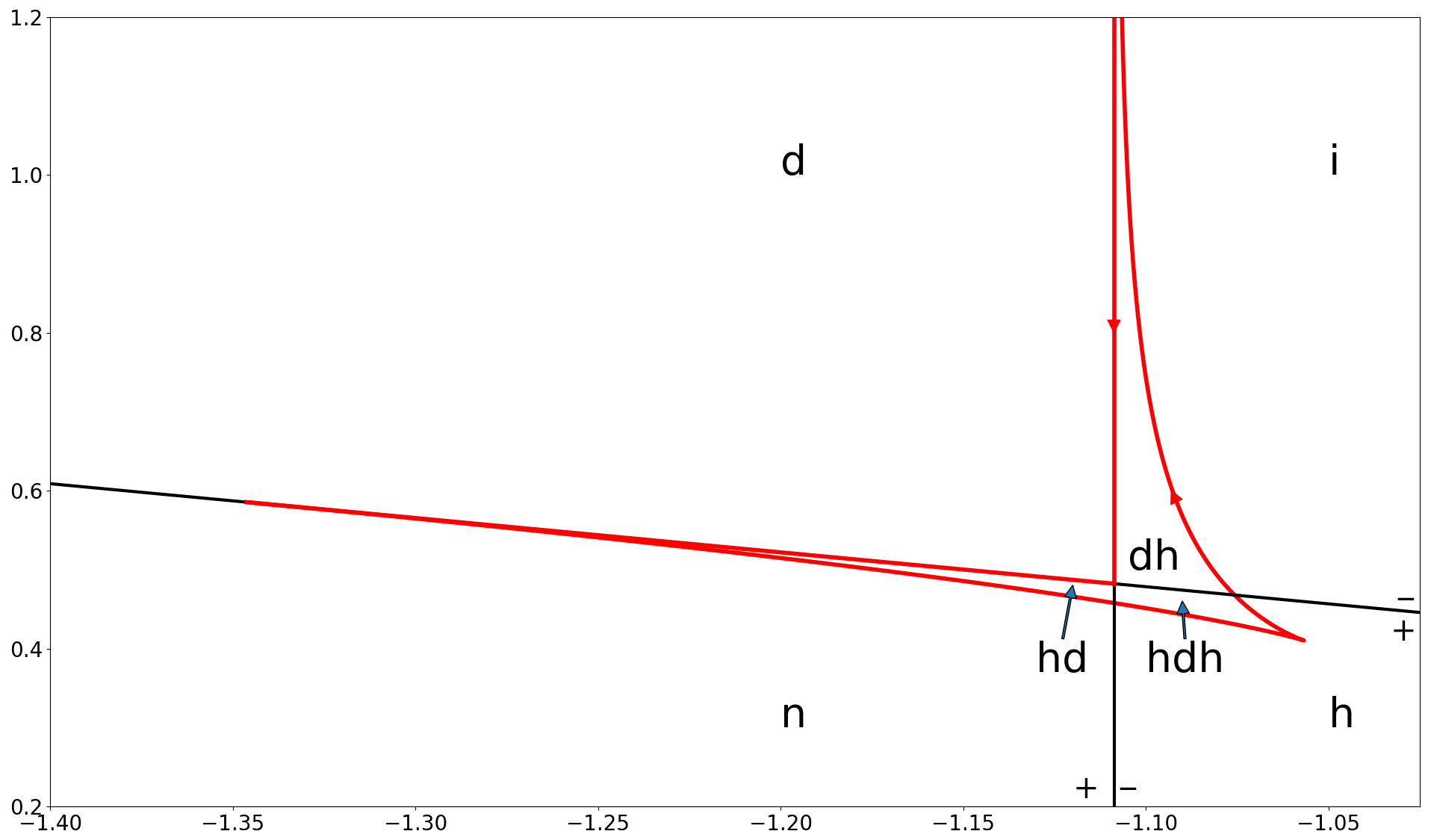}
            \caption{\label{fig_ss}Example of augmented envelope (red) and state-space decomposition for the forward curve in the scale-separated case. Parameters used are:   $\rho = 0.5$, $\lambda_1 = 1.0$, $\lambda_2 = 2.3$, $\sigma_1 = 1.0$, $\sigma_2 = 0.5$, $\theta_1 = 0.0$, $\theta_2 = 0.0$. Note that the augmented envelope contains an asymptotic point along the $z_2$-axis.}

        \end{figure}

        \begin{figure}
            \centering
            \includegraphics[width=1\textwidth]{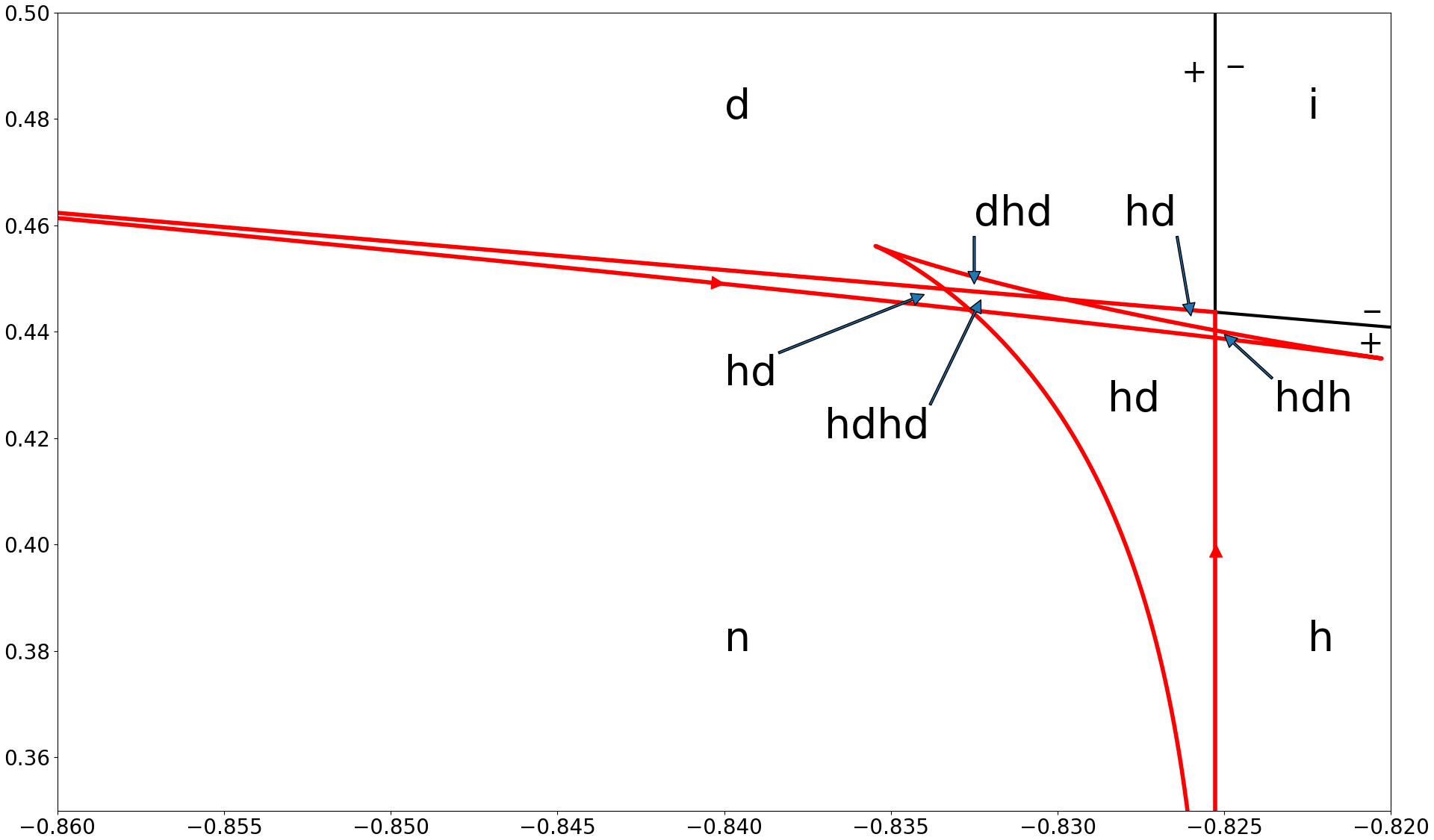}
            \caption{\label{fig:sp}Example of augmented envelope (red) and state-space decomposition for the forward curve in the scale-proximal case with $\rho < 0$. Parameters used are: $\rho = -0.5$, $\lambda_1 = 1.0$, $\lambda_2 = 1.86$,$\sigma_1 = 1.0$, $\sigma_2 = 0.65$, $\theta_1 = 0.0$, $\theta_2 = 0.0$. }
        \end{figure}

    We prepare for another corollary by making the following definitions.
    \begin{defn}\begin{enumerate}[(a)]\label{def:configuration}
    \item The (self-)intersection points of $\eta$, $\ell_0$ and $\ell_\infty$, including the contact points $\eta(0)$ and $\eta(\infty)$, are called \emph{special points} of the plane $\mathcal{Z} = \RR^2$. 
\item We call the configuration of $\eta$, $\ell_0$ and $\ell_\infty$ \emph{regular} if all special points are distinct. 
\item Adding a point $\infty$ at infinity, the plane $\mathcal{Z}$ becomes the Riemann sphere $\hat{\mathcal{Z}}$. We define the \textbf{configuration graph} as the embedded graph\footnote{Or `map', in the sense of \cite{lando2004graphs}} in $\hat{\mathcal{Z}}$, whose vertices are the special points and $\infty$, whose edges are all arcs of $\eta$, $\ell_0$ and $\ell_\infty$ connecting these points, and whose faces are the regions $R_1, R_2, \dotsc, R_K$ of Theorem~\ref{thm:split}.
 \end{enumerate}
\end{defn}
\begin{rem}
The additional point $\infty$ can be considered as the common endpoint of the rays $r_{nh}, r_{hi} ,r_{id}$ and $r_{dn}$; it is also identified with $\eta(\infty)$, if $\eta(\infty)$ is asymptotic.
\end{rem}    
    
 \begin{cor}\label{cor:configuration}
 Assume that the configuration of $\eta, \ell_0$ and $\ell_\infty$ is regular. Let $s$ be the number of self-intersections of $\eta$ and let $q_i$ be the number of intersections with $\ell_i$ where $i \in \set{0,\infty}$. Then the number $K$ of regions in Theorem~\ref{thm:split} is given by
 \begin{equation}\label{eq:count}
K = 5 + s+ q_0 + q_\infty \le 10.
\end{equation}
Exactly four of these regions correspond to the basic shapes $\texttt{normal}, \texttt{dipped}, \texttt{humped}$ and $\texttt{inverse}$. 
\end{cor}
 \begin{rem}Together with Lemma~\ref{lem:singular} we obtain a bound of $5$ regions for the scale-proximal case with $\rho \ge 0$; a bound of $7$ regions for the scale-separated case; and a bound of $10$ regions for the scale-proximal case with $\rho < 0$. Note that in the latter cases, multiple (non-adjacent) regions may correspond to the same shape.
 \end{rem}

\subsection{Proofs of the main results}
    
    The proofs of the main results will follow from a series of Lemmas. The first auxilliary result is related to the `non-zero winding rule' used in computer graphics: 
    \begin{thm}[{\cite[Prop.~3.4.4]{roe2015winding}, see also \cite{fritzsche2009grundkurs}}]\label{thm:non-zero}
        Let $\gamma:[a,b]\rightarrow \RR^2 \setminus \set{z}$ be a closed, piecewise continuously differentiable path, and let $r$ be a ray from $z$ to $\infty$, transverse to $\gamma$. Then, there are only finitely many parameter values $t = t_1, \dotsc, t_k$, where $\gamma(t)$ intersects $r$ and the winding number of $\gamma$ around $z$ is equal to 
        \[\wind_\gamma(z) = \sum_{j=1}^k i_{t_k}(\gamma, r), \]
where 
\[i(\gamma, r) = \begin{cases}+1 \quad &\text{if $r$ crosses $\gamma$ from left to right,}\\-1 \quad &\text{if $r$ crosses $\gamma$ from right to left.}\end{cases}\]
    \end{thm}
    \begin{rem}The result is shown in \cite{roe2015winding} for polygonal and for differentiable paths; it is immediate that it also holds for piecewise differentiable paths.
    \end{rem}

 As a technical tool, we will extend the augmented envelope by attaching `whiskers' at certain points:   
    
\begin{defn}The \textbf{augmented envelope with whiskers}, $\hat \eta_\times$, is an oriented curve with basepoint $M$, derived from the augmented envelope $\hat \eta$ by attaching four `whiskers' (semi-infinite line segments, traversed from contact point to infinity and back). It is defined piecewise as follows:
\begin{itemize}
\item[(w1)] A whisker extending from $M$ along $\ell_0$ and back, not containing $\eta(0)$;
\item[(a)] The line segment from $M$ to $\eta(0)$;
\item[(w2)] A whisker extending from $\eta(0)$ along $\ell_0$ and back, not containing $M$;
\item[(b)] The envelope $\eta$;
\item[(w3)] A whisker extending from $\eta(\infty)$ along $\ell_\infty$ and back, not containing $M$;
\item[(c)] The line segment from $\eta(\infty)$ to $M$;
\item[(w4)] A whisker extending from $M$ along $\ell_\infty$ and back, not containing $\eta(\infty)$.
\end{itemize}
See also \cref{fig:aeww} for an illustration.
\end{defn}
\begin{rem}
Note that Whisker (w3) can be empty if $\eta(\infty)$ is an asymptotic point.
\end{rem}
    
    We introduce some further definitions and notation associated to the augmented envelope with whiskers:  Let $z \not \in \hat \eta_\times$ be a reference point and let $h$ be an oriented helper line trough $z$ in general position and oblique to $\cF$ (see \cref{sec:helper} for the relevant definitions). We enumerate the intersection points of $h$ and $\hat \eta_\times$ in the order in which they appear on $\hat \eta_\times$, when starting from $M$, and obtain a finite sequence of intersection points $E_i$; by convention we put intersection points with $\ell_0$ and with $\ell_\infty$ in brackets:
        \[(E_0, E_1), E_2, \dots, E_{n-1}, (E_n).\]
    To each point we associate the letter $\tL$ when the point is to the left of $z$ (in the sense of \cref{def:leftright}) and the letter $\tR$ when the point is to the right of $z$. This gives a sequence of letters $\tL/\tR$, such as 
    \[(\tR \tR) \tL \tR \tR \tR \tL \tR (\tR).\]
    The sequence can be decomposed into \emph{blocks} of contiguous letters (brackets are ignored). A block is an \emph{inner block} if it is bordered by other blocks on both sides. A block is an \emph{outer block} if it is not an inner block. In the sequence above, the blocks are $\tR \tR$, $\tL$, $\tR \tR \tR$, $\tL$, $\tR \tR$.  Of these, $\tL$, $\tR \tR \tR$ and $\tL$ are inner blocks. We record some simple properties of the $\tL/\tR$-sequence:

    \begin{lem}\label{lem:odd_even}Consider the $\tL/\tR$-sequence relative to a helper line $h$ in oblique and general position.
    
    \begin{enumerate}[(a)]
    \item The $\tL/\tR$-sequence is never empty; in particular it contains two non-empty brackets.
    \item Brackets contain letters of a single kind only.
    \item The length of the $\tL/\tR$-sequence is always even.
    \item All inner blocks of the $\tL/\tR$-sequence are of odd length.
    \item The role of $\tR$'s and $\tL$'s in the sequence can be switched by reversing the orientation of the helper line $h$.
    \end{enumerate}
    \end{lem}
    \begin{proof}The helper line $h$ intersects both $\ell_0$ and $\ell_\infty$; Since the augmented envelope with whiskers $\hat \eta_\times$ transverses all points of $\ell_0$ and $\ell_\infty$ it must also intersect $h$ at least twice. Since these intersection points are with $\ell_0$ and $\ell_\infty$, they appear inside of brackets, showing (a). All intersection points in the first bracket correspond to the (unique) intersection of $h$ with $\ell_0$ and must therefore be on the same side ($\tL$ or $\tR$) of $z$. The same applies to the second bracket, showing (b). The augmented envelope with whiskers $\hat \eta_\times$ is a closed curve (up to asymptotic points, cf. \cref{rem:asymp}) and $h$ is in general position, i.e. all intersections are transverse and not tangential. Therefore the total number of intersections is even, showing (c). Consider an inner block $b$ of \tR's, which, by definition, must be bordered by blocks of \tL's. By (b) the inner block contains no intersections with $\ell_0$ or $\ell_\infty$, i.e. all intersection point in $b$ are intersections of the envelope $\eta$ with $h$ happening at distinct times $x_1 < x_2 < \dotsm < x_{|b|}$ in $(0,\infty)$. Following the tracking function $\tau$ along the helper line $h$ we deduce from \cref{lem:tracking} that $\tau'(x_i) = 0$ for each $i \in \set{1, \dotsc, |b|}$, but alternates between strictly positive and strictly negative sign between the points. Since the block $b$ is entered from the left and exited to the left, we have $\tau'(x_1 - \epsilon) > 0$ and $\tau'(x_{|b|} + \epsilon) < 0$ for $\epsilon > 0$ small enough. Applying the intermediate value theorem to the continuous function $\tau'$ it follows that the length $|b|$ of the block is odd, showing (d). Assertion (e) is trivial.
    \end{proof}
    
    The next Lemma connects the $\tL/\tR$-sequence to the extremal points of the forward or yield curve:
    \begin{lem}\label{lem:extrema}
    Let $E(z)$ be the number of local extrema of the forward curve $[0,\infty) \ni x \mapsto f(x,Z_t)$ or the yield curve $[0,\infty) \ni x \mapsto y(x,Z_t)$ in the two-dimensional Vasicek model, conditional on $Z_t = z \in \RR^2 \setminus \hat \eta$. Consider the $\tL/\tR$-sequence relative to a helper line $h$ through $z$ in oblique and general position. Then the number of blocks in the sequence is equal to $E(z) + 1$.
    \end{lem}
    \begin{proof}Consider the case of the forward curve. The equation
    \[\partial_x f(x,z) = F(x,z) = a(x) + b(x)z_1 + c(x) z_2 = 0,\]
    defines, for each $x \in (0,\infty)$, the line $\ell_x$ in the state space $Z = \RR^2$. Thus, the forward curve $x \mapsto f(x,z)$ has a critical point at some $\xi \in (0,\infty)$ if and only if $z \in \ell_\xi$. Putting a helper line in oblique and general position through $z$, this means that the number of critical points of $x \mapsto f(x,z)$ is equal to the zeroes of the tracking function $\tau$. Let $\xi$ be such a zero of $\tau$, i.e. $\tau(\xi) = 0$. Since $z \not \in \hat \eta$, we have that $\tau'(\xi) \neq 0$, by \cref{lem:tracking}. Therefore, each critical point of $x \mapsto f(x,z)$ is also an extremal point and each zero of $\tau$ is a transverse (non-tangential) zero.  Finally, note that each block of $\tR$'s corresponds to a non-empty interval $I_+ \subset [0,\infty)$ where $\tau(x) > 0$, and each block of $\tL$'s to a non-empty interval $I_- \subset [0,\infty)$ where $\tau(x) < 0$. By the intermediate value theorem, the number of zeroes of $\tau$ is equal to the number of blocks, reduced by one.
    \end{proof}
    
    \begin{lem}\label{lem:block-parity-count}
        Let $S$ be a $\tL/\tR$-sequence relative to a point $z \in \RR^2 \setminus \hat \eta$. The following is true for the winding number of $\hat \eta$ around $z$:
        \begin{enumerate}
            \item The absolute value of the winding number of $\hat \eta$ around $z$ is equal to the number of odd-length blocks of \tL's.
            \item The absolute value of the winding number of $\hat \eta$ around $z$ is equal to the number of odd-length blocks of \tR's.
        \end{enumerate}
    \end{lem}
    \begin{proof}
The helper line $h$ cuts $Z = \RR^2$ into two half-spaces and we can name them $h^+$ and $h^-$ where $h^-$ shall lie to the right of $h$, meaning to the right  when traversing $h$ along its direction vector $v$. We can associate to $\hat \eta_\times$ a sequence of numbers $\pmOne$, where $\pOne$ indicates that $\hat \eta$ crosses $h$ from the right (i.e. from $h^-$ to $h^+$) and $\mOne$ indicates the opposite direction. Note that $\hat \eta_\times$ must cross $h$ (and cannot just touch $h$) at each intersection point, since $h$ is assumed to be in general position. Since $\hat \eta_\times$ is continuous, it must cross $h$ from alternating sides, that is the sequence of $\pmOne$s must be alternating. We are left with two sequences of equal length, one consisting of $\tL/\tR$, the other of alternating $\pmOne$. In the case of \cref{fig:crossings}, these sequences are
        \begin{center}
        \begin{tabular}{cccc}
             $\tR$ & $\tL$ & $\tR$ & $\tR$ \\
             $\pOne$ & $\mOne$ & $\pOne$ & $\mOne$
        \end{tabular}
        \end{center}
        \begin{figure}
            \centering
            \includegraphics[width=\textwidth]{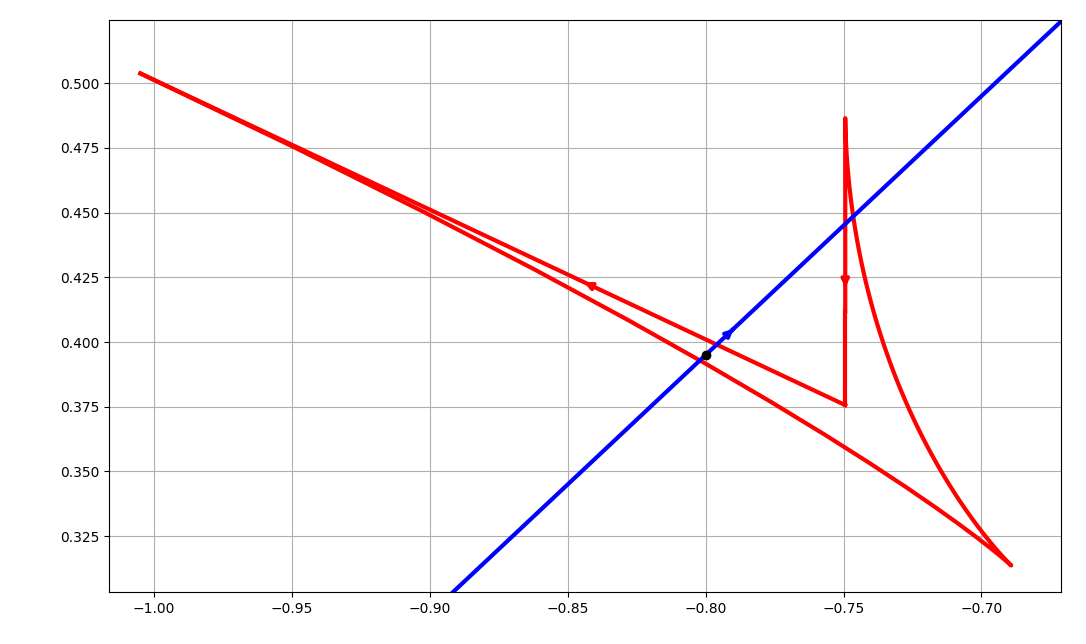}
            \caption{The augmented envelope crossing a line}
            \label{fig:crossings}
        \end{figure}
        The \emph{contribution} of a block in the $\tL/\tR$-sequence is the sum of the associated $\pmOne$'s. An odd-length block can contribute $\pOne$ or $\mOne$, and even-length block contributes $0$. Let $w$ be the winding number of $\hat \eta_\times$ around $z$. Shooting a ray from $z$ along $h$ to the right or left we see that by \cref{thm:non-zero} the total contribution of $\tR$-blocks must be equal to $w$, and the total contribution of $\tL$-blocks must be equal to $-w$ or vice versa. To complete the proof, we make the following observations: Even-length blocks contribute $0$ and can be dropped from the sequence. By \cref{lem:odd_even} all inner blocks are odd-length. Thus, dropping even-length blocks will not affect the length or structure of any other blocks. We are left with a sequence of alternating odd-length $\tL$- and $\tR$-blocks. The contribution of any $\tR$-block must be of the same sign; the same is true for $\tL$-blocks with regards to the opposite sign. It follows that the total contribution of $\tR$-blocks is equal to $+w$ or $-w$ and (a) is shown. By symmetry, the same is true for $\tL$-blocks, which shows (b).
  \end{proof}
    
 \begin{lem}\label{lem:forbidden_crossing}
 Let $h$ be a horizontal helper line. If the augmented envelope $\hat \eta$ crosses $h$ from below in a point $z$, then the next intersection with $h$ must be to the left of $z$. Conversely, if it crosses from above, then the next intersection must be to the right.
 \end{lem}
 \begin{proof}By \cref{lem:vasicek} $h$ is in oblique position. Since the statement of the Lemma is unaffected by translation of the $(z_1, z_2)$-plane, we may assume without loss of generality that $z = 0$. Moreover, we may assume that the tracking function $\tau(x)$ of $h$ is given relative to $z = 0$, i.e., $h \cap \ell_x = \tau(x) \begin{psmallmatrix}1 \\ 0 \end{psmallmatrix}$ for all $x \in (0,\infty)$. Let $x_* \in [-1,\infty)$ be the time of the  first intersection, where we use $x_* \in [-1,0]$ to parameterize the `augmented part' $\hat \eta \cap \ell_0$. Clearly, we have $\eta(x_*) = z = 0$.  Let $x^* > \max(0,x_*)$ be small enough, such that no further intersection of $\hat \eta$ with $h$ takes place in $(x_*, x^*)$ and such that $\eta(x^*)$ is not a point of regression. In addition to $z = 0$, we consider the following two points in the $(z_1, z_2)$-plane:
 \[p := \eta(x^*), \qquad q := \tau(x^*) \begin{psmallmatrix}1 \\ 0 \end{psmallmatrix} = h \cap \ell_{x^*}.\]
 Both $p$ and $q$ are located on $\ell_{x^*}$. Moreover, $\eta$ is tangent to $\ell_{x^*}$ in $p$, such that the slope of $\vec{pq}$ is equal to the slope of the tangent vector $\eta'(x^*)$. By \cref{lem:vasicek} the slope of $\eta'$ is decreasing. Together with Cauchy's mean-value theorem, we obtain
 \[\frac{\eta_2(x^*)}{\eta_1(x^*)} = \frac{\eta_2(x^*) - \eta_2(x_*)}{\eta_1(x^*) - \eta_1(x_*)} > \frac{\eta'_2(x^*)}{\eta'_1(x^*)},\]
 where the right hand side is equal to the slope of $\vec{pq}$. Taking reciprocals, we obtain
 \[\frac{\eta_1(x^*)}{\eta_2(x^*)} < \frac{\eta'_1(x^*)}{\eta'_2(x^*)} = \frac{p_1 - q_1}{p_2 - q_2} = \frac{\eta_1(x^*) - \tau(x^*)}{\eta_2(x^*)}.\]
 We note that $\eta_2(x^*) > 0$, since $p$ is located above $h$. Rearranging yields $\tau(x^*) < 0$. Together with $\tau(x_*) = 0$ this shows that the next extremal point of $\tau$ (if it exist) must take place to the left of $z=0$, i.e., for the first $\xi > x_*$ with $\tau'(\xi) = 0$ it must hold that $\tau(\xi) < 0$. By \cref{lem:tracking}, the extremal points of $\tau$ are precisely the intersection points of $h$ and $\hat \eta$, showing the claim in the first case. The second case (crossing from above) can be treated analogously.   \end{proof}

    \begin{proof}[Proof of \cref{thm:main}]
    Let $z \not \in \hat \eta_\times$. By \cref{lem:vasicek}, the horizontal line $h$ through $z$ is oblique. Typically, it will also be in general position. If it is not, then an arbitrarily small perturbation of its slope will put it in general position, since the regression points of $\hat \eta_\times$ (which have to be avoided) are isolated. Thus, we assume that the horizontal $h$ is in general position and ignore the minor adaptations that have to be made if it is not. Consider the $\tL/\tR$-sequence $S$ relative to $h$, which is non-empty and of even length $L$ by \cref{lem:odd_even}. Let $K$ be the number of blocks in $S$. 
    \begin{itemize}
    \item If $K = 1$ then $x \mapsto f(x,z)$ has no extremal points, i.e. $E(z) = 0$, by \cref{lem:extrema}. The single block of $S$ has even length, so $\wind_{\hat \eta}(z) = 0$ by \cref{lem:block-parity-count}. Moreover, the first and the last letter of $S$ are the same; therefore $h$ intersects $\ell_0$ and $\ell_\infty$ on the same side of $z$. We conclude that $z \in Q_n \cup Q_i$, and \eqref{eq:main_formula} holds true. 
    \item If $K = 2$ then $x \mapsto f(x,z)$ has a single extremal point, i.e. $E(z) = 1$, by \cref{lem:extrema}. The first and the last letter of $S$ are different; therefore $h$ intersects $\ell_0$ and $\ell_\infty$ on different sides of $z$ and $z \in Q_d \cup Q_h$. Since the total length of $S$ is even, the lengths of the two blocks of $S$ are even-even or odd-odd. In the first case $\wind_{\hat \eta}(z) = 0$ by \cref{lem:block-parity-count}, in the second case $\wind_{\hat \eta}(z) = 1$. Thus, \eqref{eq:main_formula} holds true if the odd-odd case can be excluded.
    \end{itemize}
    If $K > 2$ then the sequence $S$ contains $K - 2 > 0$ inner blocks. Again, we have to distinguish odd and even $K$:
    \begin{itemize}
    \item If $K = 2J + 1$ with $J \in \NN$ then $E(z) = 2J$, by \cref{lem:extrema}. All $2J - 1$ inner blocks of $S$ are of odd length (see \cref{lem:odd_even}) and the total length of $S$ is even. Therefore the two outer blocks have lengths odd-even or even-odd. In any case the total number of odd-length blocks is $2J$ and $\wind_{\hat \eta}(z) = J$ by \cref{lem:block-parity-count}. Moreover, the first and the last letter of $S$ are the same; therefore $z \in Q_n \cup Q_i$, and \eqref{eq:main_formula} holds true. 
    \item If $K = 2J + 2$ with $J \in \NN$ then $E(z) = 2J + 1$, by \cref{lem:extrema}. The number of blocks is even, therefore the first and the last letter of $S$ are different and $z \in Q_d \cup Q_h$. All $2J$ inner blocks of $S$ are of odd length (see \cref{lem:odd_even}) and the total length of $S$ is even. Therefore the two outer blocks have lengths even-even or odd-odd. In the first case $\wind_{\hat \eta}(z) = J$ by \cref{lem:block-parity-count}, in the second case $\wind_{\hat \eta}(z) = J+1$. Thus, \eqref{eq:main_formula} holds true if the odd-odd case can be excluded.
    \end{itemize}
    To complete the proof, it remains to show that there can not be two outer blocks in $S$ which are both of odd length. By the parity arguments given above, this can only happen if $z \in Q_d \cup Q_h$. Assume that $z \in Q_d$; the case $z \in Q_h$ can be treated symmetrically. We orient the helper line from left to right (in the $(z_1, z_2)$-plane) such that the sequence $S$ starts with $\tL$ (corresponding to the intersection with $\ell_0$). As in the proof of \cref{lem:block-parity-count} we associate an alternating sequence of $\mOne/\pOne$ to the sequence $S$, corresponding to whether $h$ is crossed from the left or from the right. Since the base point $M$ of $\hat \eta_\times$ lies below $h$, this sequence has to start with $\pOne$, corresponding to a crossing from the right. Since the length of the first block is odd, it must end with the pattern
        \begin{center}
        \begin{tabular}{cccc}
$\dotsm$            $\tL$ & $\tR$ $\dotsm$\\
$\dotsm$              $\pOne$ & $\mOne$ $\dotsm$.
        \end{tabular}
    \end{center}
Since intersection points with whiskers come in pairs, both of the intersections represented by the above pattern correspond to intersections of the augmented envelope (without whiskers) $\hat \eta$ and $h$. The first intersection crosses $h$ from below and the second intersection takes place to the right of the first. This contradicts \cref{lem:forbidden_crossing}. Thus, the odd-odd case is ruled out, completing the proof. 
    \end{proof}

We complete the section with the proof of our second main result and its corollary.
    
 \begin{proof}[Proof of \cref{thm:split}]
 Consider the configuration graph $\cG$ (see \cref{def:configuration}) induced by $\eta, \ell_0$ and $\ell_\infty$. It has finitely many vertices of finite degree; therefore also the number of edges and the number of faces, i.e., of regions $R_1, \dotsc, R_K$ is finite. The winding number is constant on every connected component of $\RR^2 \setminus \hat \eta$, see e.g. \cite[Ch.~5.7]{do2016differential}. These connected components, possibly split by $\ell_0$ or $\ell_\infty$, are exactly the regions $R_1, \dotsc, R_K$. Together with \cref{thm:main} it follows that the shape of the forward- or yield-curve is an invariant of each such region $R_i$. If two distinct regions $R_i$ and $R_j$ share a boundary containing a one-dimensional submanifold, then this submanifold must be part of $\eta$, of $\ell_0$, or of $\ell_\infty$. In any case, \eqref{eq:main_formula}, together with \cref{thm:non-zero} shows that the shapes in $R_i$ and $R_j$ must be distinct. 
 \end{proof}
 \begin{proof}[Proof of \cref{cor:configuration}]
 Consider the configuration graph $\cG$ of $\eta, \ell_0$ and $\ell_\infty$,  embedded into the Riemann sphere. Assume first that $\eta(\infty)$ is non-asymptotic. The total number $v$ of vertices in $\cG$ is
\[v = 4 + s + q_0 + q_\infty,\]
where the four vertices that are always present are $M, \eta(0), \eta(\infty)$ and $\infty$. The degree of all vertices is $4$; except for $\eta(0)$ and $\eta(\infty)$, which have degree 3. Thus, the number $e$ of edges of $\cG$, which is half the total degree, is 
\[e = 7+ 2 (s + q_0 + q_\infty).\]
By the Euler-Poincar\'e formula (see e.g. \cite[Thm.~1.3.10]{lando2004graphs}) $v - e + K = 2$, where $2$ is the Euler characteristic of the sphere and $K$ the number of faces.  We conclude that
\[K = 2 + e - v = 5 + s + q_0 + q_\infty,\]
showing \cref{eq:count}. Finally, the winding number of $\hat \eta$ is zero a neighborhood  of $\infty$ in $\hat{\mathcal{Z}}$. This neighborhood intersects all quadrants $Q_n, Q_d, Q_i$ and $Q_h$. Thus, by \cref{thm:main} the shapes in the corresponding regions are $\texttt{normal}, \texttt{dipped}, \texttt{humped}$ and $\texttt{inverse}$.\\
If $\eta(\infty)$ is asymptotic, then the number of vertices is reduced by one. However, the total degree is also decreased by two, such that  \cref{eq:count} remains unchanged. The point $\infty$ borders five regions in this case, but four of them must still correspond to the shapes $\texttt{normal}, \texttt{dipped}, \texttt{humped}$ and $\texttt{inverse}$.
 \end{proof}
    
\section{Consequences and further results}\label{sec:shape}
We now look at several consequences from the results obtained in the previous section. Our first object of interest are the possible \emph{transitions} between different shapes of the forward- or yield-curve. 
\begin{defn}\label{def:transition}
The \textbf{transition graph} $\mathcal{T}$ is a (labelled and embedded into $\RR^2$) graph, whose 
vertices are the regions of \cref{thm:split}, labelled by the associated shape of the forward- (or yield-) curve. Two vertices are linked if the corresponding regions are adjacent in the sense of \cref{thm:split}.
\end{defn}
It is not difficult to see that the transition graph is the \emph{dual graph} (see also \cite{lando2004graphs}) of the configuration graph from \cref{def:configuration}. It is also clear that we can define an equivalence relation on the parameter space of the two-dim.{} Vasicek model, where we call two parameter tuples equivalent if they produce the same\footnote{In the sense of isomorphisms of embedded graphs, cf. \cite{lando2004graphs}} transition graphs (and consequently also configuration graphs). For such equivalent tuples the state-space picture of forward- (or yield-)curve shapes and their transitions is essentially the same, i.e., the same shapes with the same transitions appear. Here, we attempt to classify all possible equivalence classes and exhibit some of the associated transition graphs. In most cases the following strategy is sufficient to uniquely determine the transition graph of a certain parameter set: 
\begin{itemize}
\item Determine whether $\eta(0)$ is located on $r_{nd}$ or $r_{ih}$,
\item Determine the quadrant into which $\eta$ moves immediately after $x = 0$,
\item Count the number of cusps, self-intersections and $\ell_0/\ell_\infty$-intersections of $\eta$,
\item Determine last quadrant that $\eta$ visits as $x \to \infty$,
\item Determine whether $\eta(\infty)$ is located on $r_{di}$ or $r_{hn}$.
\end{itemize}
Some rare edge cases may appear, when the configuration is not regular in the sense of \cref{def:configuration}, for example when $\eta(0)$ or $\eta(\infty)$ coincides with the midpoint $M$.

\subsection{Transitions between shapes of the forward curve}    \label{sec:fw_trans}
In the forward case, the starting point of the envelope is given by
\[\eta_{\rm f}(0) = \left(\theta_1 + \frac{\sigma_1^2 + \sigma_2^2 + 2 \rho \sigma_1 \sigma_2}{\lambda_1(\lambda_1 - \lambda_2)}, \theta_2 + \frac{\sigma_1^2 + \sigma_2^2 + 2 \rho \sigma_1 \sigma_2}{\lambda_2(\lambda_2 - \lambda_1)}\right).\]
The endpoint $\eta_{\rm f}(\infty)$ is as follows 
\begin{itemize}
\item Scale-proximal case: $\eta_{\rm f}(\infty) = \left(\theta_i - \tfrac{\sigma_1^2}{\lambda_1^2} - \rho\tfrac{\sigma_1\sigma_2}{\lambda_1\lambda_2}, \theta_2 - \tfrac{\sigma_2^2}{\lambda_2^2} - \rho\tfrac{\sigma_1\sigma_2}{\lambda_1\lambda_2}\right)$
\item Scale-separated case: $\eta_{\rm f}(\infty)$ is asymptotic (along $\ell_\infty$).
\end{itemize}
The direction of the tangent vector $\eta'(x)$ as $x \to 0$ and $x \to \infty$ can be determined by \cref{lem:envelope_slope}.
Comparing $\eta(0)$ and $\eta'(0)$ relative to $\ell_\infty$ and $\eta(\infty)$ and $\eta'(\infty)$ relative to $\ell_0$ it is straight-forward to show the following results:
    \begin{lem}\label{lem:fw_initial}
        The envelope starts on $r_{nd}$ iff
        \begin{align}\label{eqn:start_comp}
            \sigma_1^2\left(\frac{2\lambda_1-\lambda_2}{\lambda_1}\right) + \rho\sigma_1\sigma_2\left(\frac{\lambda_1 + \lambda_2}{\lambda_2}\right) + \sigma_2^2 \ge 0
        \end{align}
        and on $r_{ih}$ iff the inequality is reversed. It moves on into $Q_n$ or $Q_h$ iff 
             \begin{align}\label{eqn:dir_start}
            \sigma_1^2(2\lambda_1 - \lambda_2) + \rho\sigma_1\sigma_2(\lambda_1 + \lambda_2) + \sigma_2^2(2\lambda_2 - \lambda_1) > 0,
        \end{align}
        and into $Q_d$ or $Q_i$ if the inequality is reversed. In the scale-proximal case with $\rho \ge 0$ the inequalities \eqref{eqn:start_comp} and \eqref{eqn:dir_start} hold true for all $\sigma_1, \sigma_2 \in (0,\infty)$. 
 \end{lem}

    \begin{lem}\label{lem:fw_terminal}
    \flushleft
    \begin{enumerate}
    \item In the scale-separated case the envelope ends in $r_{id}$, coming from $Q_i$.
    \item In the scale-proximal case the envelope ends in $r_{hn}$ coming from $Q_n$, if 
    \begin{align}\label{eqn:end_comp}
        \frac{\sigma_1^2}{\lambda_1} + \rho\sigma_1\sigma_2\frac{\lambda_1 + \lambda_2}{\lambda_1\lambda_2} + \frac{\sigma_2^2}{\lambda_2} > 0,
    \end{align}
    and ends in $r_{id}$ coming from $Q_d$  if the inequality is reversed. If $\rho \ge 0$, then inequality \eqref{eqn:end_comp} holds true for all $\sigma_1, \sigma_2 \in (0,\infty)$.
\end{enumerate}
 \end{lem}
    
Using these results we give a complete classification of transition graphs in the scale-proximal case with $\rho \ge 0$ and in the scale-separated case.
\subsubsection{The scale-proximal case with $\rho \ge 0$}    
In this case, \cref{lem:fw_initial} and \cref{lem:fw_terminal} imply that $\eta$ starts on $r_{nd}$, moving into $Q_n$, and ends on $r_{hn}$, coming from $Q_n$. By \cref{lem:singular} there are no self-intersections and no intersections with $\ell_0$ or $\ell_\infty$, such that $\eta$ never leaves $Q_n$. The unique transition graph is shown in \cref{fig:LSCG1} and an example of the state space decomposition is shown in \cref{fig:sp0}.

\subsubsection{The scale-separated case}
Here, by \cref{lem:fw_terminal}, the envelope $\eta$ always ends in $r_{di}$ coming from $Q_i$. The contact point $\eta(\infty)$, however, is an asymptotic point. By \cref{lem:singular} $\eta$ can have at most a single cusp and at most a single intersection point with each $\ell_0$ and $\ell_\infty$. Depending on the start of $\eta$ we distinguish the following cases; the corresponding transition graphs are shown in \cref{fig:LSCG2}.
\begin{enumerate}
\item The envelope starts in $r_{ih}$, moving into $Q_i$. In this case there can be no cusp, no self-intersections and no intersections with $\ell_0$ or $\ell_\infty$. Any such points, together with the increasing-slope property would lead to a contradiction. Thus, $\eta$ stays in $Q_i$ for all times. 
\item The envelope starts in $r_{ih}$, moving into $Q_h$. In this case $\eta$ must have a cusp, located in $Q_h$, and an intersection with $\ell_0$, but no other special points. A possible edge case appears if $\eta(0)$ coincides with $M$. 
\item The envelope starts in $r_{nd}$, moving into $Q_n$. In this case $\eta$ must have a cusp, located in $Q_h$, and an intersection with both $\ell_0$ and $\ell_\infty$, but no other special points. An example of this state-space decomposition is shown in \cref{fig_ss}.
\end{enumerate}

The case of starting in $r_{nd}$ and moving into $Q_d$ can be ruled out; together with the end on $r_{di}$ it poses a contradiction to the combination of decreasing slope and single cusp. Thus the list of cases is complete. Note that all cases can be attained by choosing a suitable ratio of $\sigma_2/\sigma_1$ in the inequalities \eqref{eqn:start_comp} and \eqref{eqn:end_comp}.
    
\subsubsection{Other cases}
The scale-critical cases can be analyzed in a similar way, with a complete classification of all subcases given in \cref{tab:classification}. The scale-proximal case with $\rho < 0$ is the most difficult case; \cref{tab:classification} lists all subcases that are possible a priori, i.e., that do not lead to any contraction to the restrictions imposed by Lemmas~\labelcref{lem:envelope_slope}, \labelcref{lem:singular}, \labelcref{lem:fw_initial}, and \labelcref{lem:fw_terminal}. However, we were not able to analytically show that these cases are indeed attainable, i.e., to guarantee the existence of parameter tuples that lead to the corresponding transition graphs and state-space decompositions. Some cases are easy to confirm numerically, for example a state space decomposition corresponding to the last row of \cref{tab:classification} is shown in \cref{fig:sp}. It remains an open problem to completely classify (including attainability) all transition graphs in the scale-proximal case with $\rho < 0$.

        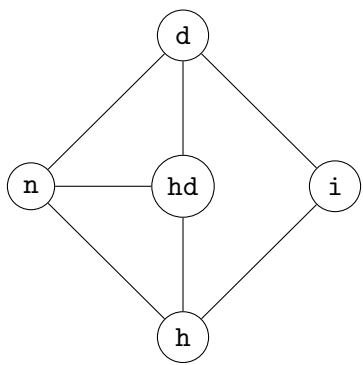
\begin{figure}
        \centering
        \begin{tikzpicture}[scale=2]
        	\node (A) at (0,0) [circle,draw] {\texttt{hd}};
        	\node (B) at (-1,0) [circle,draw] {\texttt{n}};
        	\node (C) at (0,1) [circle,draw] {\texttt{d}};
        	\node (D) at (1,0) [circle,draw] {\texttt{i}};
        	\node (E) at (0,-1) [circle,draw] {\texttt{h}};
        	
        	\draw[-] (A) to (B);
        	\draw[-] (A) to (C);
        	\draw[-] (A) to (E);
        	\draw[-] (B) to (C);
        	\draw[-] (C) to (D);
        	\draw[-] (D) to (E);
        	\draw[-] (E) to (B);
        \end{tikzpicture}
        \caption{Transition graph for the scale proximal case with $\rho \ge 0$}
        \label{fig:LSCG1}
    \end{figure}
    
    \begin{figure}
        \centering
        \begin{subfigure}{.45\linewidth}
            \begin{tikzpicture}[scale=1.7]
            	\node (A) at (0,0) [circle,draw] {\texttt{dh}};
            	\node (B) at (-1,0) [circle,draw] {\texttt{n}};
            	\node (C) at (0,1) [circle,draw] {\texttt{d}};
            	\node (D) at (1,0) [circle,draw] {\texttt{i}};
            	\node (E) at (0,-1) [circle,draw] {\texttt{h}};
            	
            	\draw[-] (A) to (D);
            	\draw[-] (A) to (C);
            	\draw[-] (A) to (E);
            	\draw[-] (B) to (C);
            	\draw[-] (D) to (E);
            	\draw[-] (E) to (B);
            \end{tikzpicture}
            \subcaption{Case (1)}
	    \end{subfigure}
	    \begin{subfigure}{.45\linewidth}
            \begin{tikzpicture}[scale=1.7]
            	\node (A) at (0,0) [circle,draw] {\texttt{dh}};
            	\node (B) at (-2,0) [circle,draw] {\texttt{n}};
            	\node (C) at (0,1) [circle,draw] {\texttt{d}};
            	\node (D) at (1,0) [circle,draw] {\texttt{i}};
            	\node (E) at (0,-1) [circle,draw] {\texttt{h}};
            	\node (F) at (-1,0) [circle,draw] {\texttt{hdh}};
            	
            	\draw[-] (A) to (D);
            	\draw[-] (A) to (C);
            	\draw[-, red] (A) to (E);
            	\draw[-] (B) to (C);
            	\draw[-] (D) to (E);
            	\draw[-] (E) to (B);
            	\draw[-] (A) to (F);
            	\draw[-] (F) to (E);
            \end{tikzpicture}
            \subcaption{Case (2) with edge case}
	    \end{subfigure}
	    
	    \begin{subfigure}{.45\linewidth}
            \begin{tikzpicture}[scale=1.7]
            	\node (A) at (-0.5,0) [circle,draw] {\texttt{hd}};
            	\node (B) at (-1.5,0) [circle,draw] {\texttt{n}};
            	\node (C) at (0,1) [circle,draw] {\texttt{d}};
            	\node (D) at (1.5,0) [circle,draw] {\texttt{i}};
            	\node (E) at (0,-1.5) [circle,draw] {\texttt{h}};
            	\node (F) at (0,-0.5) [circle,draw] {\texttt{hdh}};
            	\node (G) at (0.5,0) [circle,draw] {\texttt{dh}};
            	
            	\draw[-] (A) to (C);
            	\draw[-] (B) to (C);
            	\draw[-] (D) to (E);
            	\draw[-] (E) to (B);
            	\draw[-] (A) to (F);
            	\draw[-] (F) to (E);
            	\draw[-] (F) to (G);
            	\draw[-] (C) to (G);
            	\draw[-] (D) to (G);
            	\draw[-] (A) to (B);
            \end{tikzpicture}
            \subcaption{Case (3)}
	    \end{subfigure}
        \caption{Transition graphs for the scale-separated case}
        \label{fig:LSCG2}
    \end{figure}
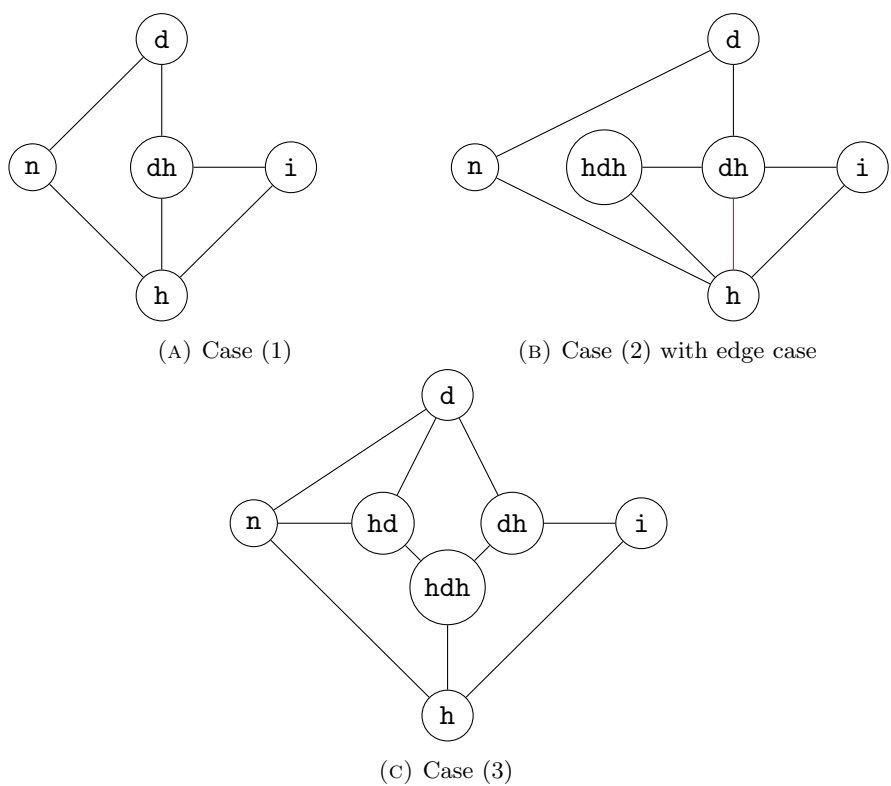

\subsection{Transitions between shapes of the yield curve}\label{sec:y_trans}
The case of the yield curve can in principle be analyzed in the same way as the forward curve case. By \cref{lem:start} we know that the number of cusps of $\eta_{\rm y}$ is less or equal than the number of cusps of $\eta_{\rm f}$, suggesting at first that no new transition graphs appear in the yield curve case. On the other hand the contact point $\eta_{\rm y}(\infty)$, which was asymptotic in the scale-separated forward curve case, is never asymptotic for the yield curve, see \cref{lem:start}(b). For this reason, two additional transition graphs \emph{do} appear in the yield curve case, one for the scale-separated case and one for the scale-critical case with $\rho < 0$, see \cref{tab:classification}. We report the analogue results of \cref{lem:fw_initial} and \cref{lem:fw_terminal}:
    
      \begin{lem}\label{lem:y_initital}
        The envelope $\eta_{\rm y}$ starts on $r_{nd}$ iff
                \begin{align*}
            &\sigma_1^2\left(\frac{4\lambda_1^2 + 4\lambda_1\lambda_2-3\lambda_2^2}{4\lambda_1^2}\right) + \rho \frac{\sigma_1\sigma_2}{\lambda_1\lambda_2}\left(\frac{\lambda_2^3 + 3\lambda_1\lambda_2^2 + \lambda_1^2\lambda_2}{\lambda_1 + \lambda_2}\right) +\\&\quad+ \sigma_2^2\left(\frac{4\lambda_2^2 + 4\lambda_1\lambda_2 - 3\lambda_1^2}{4\lambda_1\lambda_2}\right) \ge 0
        \end{align*}
        and on $r_{ih}$ iff the inequality is reversed. It moves on into $Q_n$ or $Q_h$ iff 
             \begin{align*}
            \sigma_1^2(2\lambda_1 - \lambda_2) + \rho\sigma_1\sigma_2(\lambda_1 + \lambda_2) + \sigma_2^2(2\lambda_2 - \lambda_1) > 0,
        \end{align*}
        and into $Q_d$ or $Q_i$ if the inequality is reversed. In the scale-proximal case with $\rho \ge 0$ both inequalities above hold true for all $\sigma_1, \sigma_2 \in (0,\infty)$. 
 \end{lem}
    
   \begin{lem}\flushleft\label{lem:y_terminal}
The envelope $\eta_{\rm y}$ ends in $r_{hn}$ if 
    \begin{align}
        \frac{\sigma_1^2}{\lambda_1} + \rho\sigma_1\sigma_2\frac{\lambda_1 + \lambda_2}{\lambda_1\lambda_2} + \frac{\sigma_2^2}{\lambda_2} > 0,
    \end{align}
    and ends in $r_{id}$ if the inequality is reversed. It does so coming from $Q_n$ or $Q_d$ in the scale-proximal case, and from $Q_i$ or $Q_h$ in the scale-separated case. If $\rho \ge 0$, then the inequality above holds true for all $\sigma_1, \sigma_2 \in (0,\infty)$.
 \end{lem}
        
The individual cases can now be classified by arguments analogous to the forward curve case. Without giving full details, the results are summarized in \cref{tab:classification}.

    \begin{table}
    \scriptsize
     \captionsetup{width=0.9 \textwidth}
\begin{adjustbox}{center}
\begin{tabular}{|p{2cm}||c|c|c|c||c|c|c|c||c|c|c|}
        \hline
Regime \vspace{.3em}& \rotatebox{90}{$\eta$ starts on} & \rotatebox{90}{$\eta$ moves into} & \rotatebox{90}{$\eta$ ends on} & \rotatebox{90}{$\eta$ exits from} & \rotatebox{90}{\# cusps} & \rotatebox{90}{\# self-inters.} & \rotatebox{90}{\# inters.\ w/ $\ell_0$} & \rotatebox{90}{\# inters.\ w/ $\ell_\infty$} & \rotatebox{90}{$\eta(\infty)$ asymptotic} & \rotatebox{90}{\parbox{3cm}{additional shapes\\(besides \texttt{n}, \texttt{d},\texttt{i}, \texttt{h})}} & \rotatebox{90}{\parbox{3cm}{attainable for\\forward/yield curve}}\\
        \hline
        \hline
        scale-proximal with $\rho \ge 0$ & $r_{nd}$ & $Q_n$ & $r_{nh}$ & $Q_n$ & 0 & 0 & 0 & 0 & no & \texttt{hd} & f,y\\
        \hline
        scale-critical with $\rho \ge 0$  & $r_{nd}$ & $Q_n$ & $r_{nh}$ & $Q_n$ & 0 & 0 & 0 & 0 & no & \texttt{hd} & f,y\\
        \hline
        scale-separated & $r_{ih}$ & $Q_{i}$ & $r_{id}$ & $Q_{i}$ & 0 & 0 & 0 & 0 & f & \texttt{dh} & f,y\\
         & $r_{ih} $ & $Q_{h}$ & $r_{id}$ & $Q_{i}$ & 1 & 0 & 1 & 0 & f & \texttt{dh}, \texttt{hdh} & f,y\\
         & $r_{nd}$ & $Q_{n}$ & $r_{id}$ & $Q_{i}$ & 1 & 0 & 1 & 1 & f & \texttt{hd}, \texttt{dh}, \texttt{hdh} & f,y\\
         & $r_{nd}$ & $Q_{n}$ & $r_{nh}$ & $Q_{h}$ & 1 & 0 & 0 & 1 & no & \texttt{hd}, \texttt{hdh} & y\\
         \hline
        \multirow{2}{2cm}{scale-critical with $\rho < 0$} & $r_{ih}$ & $Q_{i}$ & $r_{id}$ & $Q_{i}$ & 0 & 0 & 0 & 0 & no & \texttt{dh} & f,y\\
         & $r_{ih} $ & $Q_{h}$ & $r_{id}$ & $Q_{i}$ & 1 & 0 & 1 & 0 & no & \texttt{dh}, \texttt{hdh} & f,y\\
         & $r_{nd}$ & $Q_{n}$ & $r_{id}$ & $Q_{i}$ & 1 & 0 & 1 & 1 & no & \texttt{hd}, \texttt{dh}, \texttt{hdh} & f,y\\
         & $r_{nd}$ & $Q_{n}$ & $r_{nh}$ & $Q_{h}$ & 1 & 0 & 1 & 1 & no & \texttt{hd}, \texttt{hdh} & f,y\\
         & $r_{nd}$ & $Q_n$ & $r_{nh}$ & $Q_n$ & 0 & 0 & 0 & 0 & no & \texttt{hd} & y\\
        \hline
        \multirow{2}{2cm}{scale-proximal with $\rho < 0$} & $r_{nd}$ & $Q_{d}$ & $r_{nh}$ & $Q_{n}$ & 1 & 0 & 1 & 0 & no & \texttt{hd}, \texttt{dhd} & f/\textbf{?}\\
         & $r_{ih}$ & $Q_{i}$ & $r_{nh}$ & $Q_{n}$ & 1 & 0 & 1 & 1 & no & \texttt{hd}, \texttt{dh}, \texttt{dhd} & f/\textbf{?}\\
         & $r_{ih}$ & $Q_{i}$ & $r_{id}$ & $Q_{d}$ & 1 & 0 & 0 & 1 & no & \texttt{dh}, \texttt{dhd} & \textbf{?}/y\\
         & $r_{nd}$ & $Q_{n}$ & $r_{id}$ & $Q_{d}$ & 2 & 0 & 1 & 2 & no & \texttt{hd}, \texttt{dh}, \texttt{hdh}, \texttt{dhd} & f/\textbf{?}\\
         & $r_{ih}$ & $Q_{h}$ & $r_{id}$ & $Q_{d}$ & 2 & 0 & 1 & 1 & no & \texttt{dh}, \texttt{hdh}, \texttt{dhd} & f/y\\
         & $r_{ih}$ & $Q_{h}$ & $r_{nh}$ & $Q_{n}$ & 2 & 0 & 2 & 1 & no & \texttt{hd}, \texttt{dh}, \texttt{hdh}, \texttt{dhd} & f/\textbf{?}\\
         \cline{2-12}
         & $r_{nd}$ & $Q_n$ & $r_{nh}$ & $Q_n$ & 0 & 0 & 0 & 0 & no & \texttt{hd} & \textbf{?}\\
         & $r_{nd}$ & $Q_n$ & $r_{nh}$ & $Q_n$ & 2 & 1 & 0 & 0 & no & \texttt{hd}, \texttt{hdhd} & \textbf{?}\\
         & $r_{nd}$ & $Q_n$ & $r_{nh}$ & $Q_n$ & 2 & 1 & 2 & 0 & no & \texttt{hd}, \texttt{dhd}, \texttt{hdhd} & \textbf{?}\\
         & $r_{nd}$ & $Q_n$ & $r_{nh}$ & $Q_n$ & 2 & 0 & 0 & 2 & no & \texttt{hd}, \texttt{dhd}, \texttt{hdhd} & \textbf{?}\\
         & $r_{nd}$ & $Q_n$ & $r_{nh}$ & $Q_n$ & 2 & 1 & 0 & 2 & no & \texttt{hd}, \texttt{hdh}, \texttt{hdhd} & \textbf{?}\\
         & $r_{nd}$ & $Q_n$ & $r_{nh}$ & $Q_n$ & 2 & 0 & 2 & 2 & no & \texttt{hd}, \texttt{hdh}, \texttt{dhd}, \texttt{hdhd}$^*$ & \textbf{?}\\
         & $r_{nd}$ & $Q_n$ & $r_{nh}$ & $Q_n$ & 2 & 0 & 2 & 2 & no & \texttt{hd}, \texttt{dh}, \texttt{hdh}, \texttt{dhd}, \texttt{hdhd}$^*$ & \textbf{?}\\
         & $r_{nd}$ & $Q_n$ & $r_{nh}$ & $Q_n$ & 2 & 1 & 2 & 2 & no & \texttt{hd}, \texttt{hdh}, \texttt{dhd}, \texttt{hdhd} & \textbf{?}\\
         & $r_{nd}$ & $Q_n$ & $r_{nh}$ & $Q_n$ & 2 & 1 & 2 & 2 & no & \texttt{hd}, \texttt{dh}, \texttt{hdh}, \texttt{dhd}, \texttt{hdhd} & \textbf{?}\\
         \hline
    \end{tabular}
    \end{adjustbox}
    \caption{\label{tab:classification}This table lists, row-by-row, all possible (i.e., allowed by our results) equivalence classes of transition graphs and their occurrent shapes for the two-dim.\ Vasicek model. The letters `f'. and `y' stand for forward- and yield-curve respectively; a question mark indicates an unknown property. For rows marked with `$*$' two (non-isomorphic) transition graphs exist. Additional explanations can be found in Sections~\labelcref{sec:fw_trans} and \labelcref{sec:y_trans}.}
    \end{table}
    
\subsection{Shapes for given short-rate}    
So far we have studied the shape of the forward and the yield curve given a certain state $Z_t = z \in \RR^2$ of the factor process. Now, we want to analyze the possible shapes, given the state $r_t = r \in \RR$ of the short rate. The short rate is related to $Z$ by 
\[r_t = \kappa + Z_t^1 + Z_t^2,\]
such that fixing the short rate amounts to fixing a line of slope $-1$ in the state space $\RR^2$. From \cite[Fig. 3]{diez2020yield} we know that there are short rates for which three different shapes of the yield curve occur. We give a more refined result, showing that three is in fact a lower bound for the number of shapes given \emph{any} level of $r \in \RR$.
    \begin{thm}Consider the two-dimensional Vasicek model in regular configuration. 
        The following properties hold for the kind and number of shapes that are attainable for a given short rate, for both the forward and yield curve:
        \begin{enumerate}[(a)]
            \item For every choice of parameters and every given short rate $r_t = r \in \RR$ at least the shapes \texttt{dipped} and \texttt{humped} occur.
            \item For every choice of parameters and every short rate $r_t = r \in \RR$ at least 3 different shapes occur.
            \item For every choice of parameters there exists a short rate $r_t = r \in \RR$ for which 4 different shapes occur.
        \end{enumerate}
    \end{thm}
    \begin{proof}Fixing the short rate $r_t = r$ corresponds to choosing a line $h$ of slope $-1$ in the $(z_1,z_2)$-plane. The slopes of $\ell_0$ and $\ell_\infty$ are $(-\tfrac{\lambda_1}{\lambda_2},-\infty)$ in the forward case and $(-\tfrac{\lambda_1}{\lambda_2},-\tfrac{\lambda_2}{\lambda_1})$ in the yield case. Both intervals contain $-1$, such that $h$ intersects both $Q_h$ and $Q_d$, even on the outside of $\hat \eta$, showing (a). If it avoids the midpoint $M$ it also intersects a third quadrant and (b) is shown. If it meets $M$, then either $M \in \hat \eta$, in which case the shape at $M$ is either \texttt{n} or \texttt{i} by \cref{rem:online}, or $M \not \in \eta$, in which case $M$ borders four regions with different shapes by \cref{cor:configuration}, also showing (b). The envelope always intersects with the interior of $Q_n$ or $Q_i$ and a line with slope $-1$ can be tangent to it only at a single point. By taking some other regular point on the envelope in one of the mentioned quadrants, and running a line of slope $-1$ through it we get two distinct shapes that are not \texttt{d} or \texttt{h}, showing (c). 
    \end{proof}
We give an example which shows that even five different shapes can occur for a single given short rate.
        \begin{figure}
                \centering
                \includegraphics[width=1\textwidth]{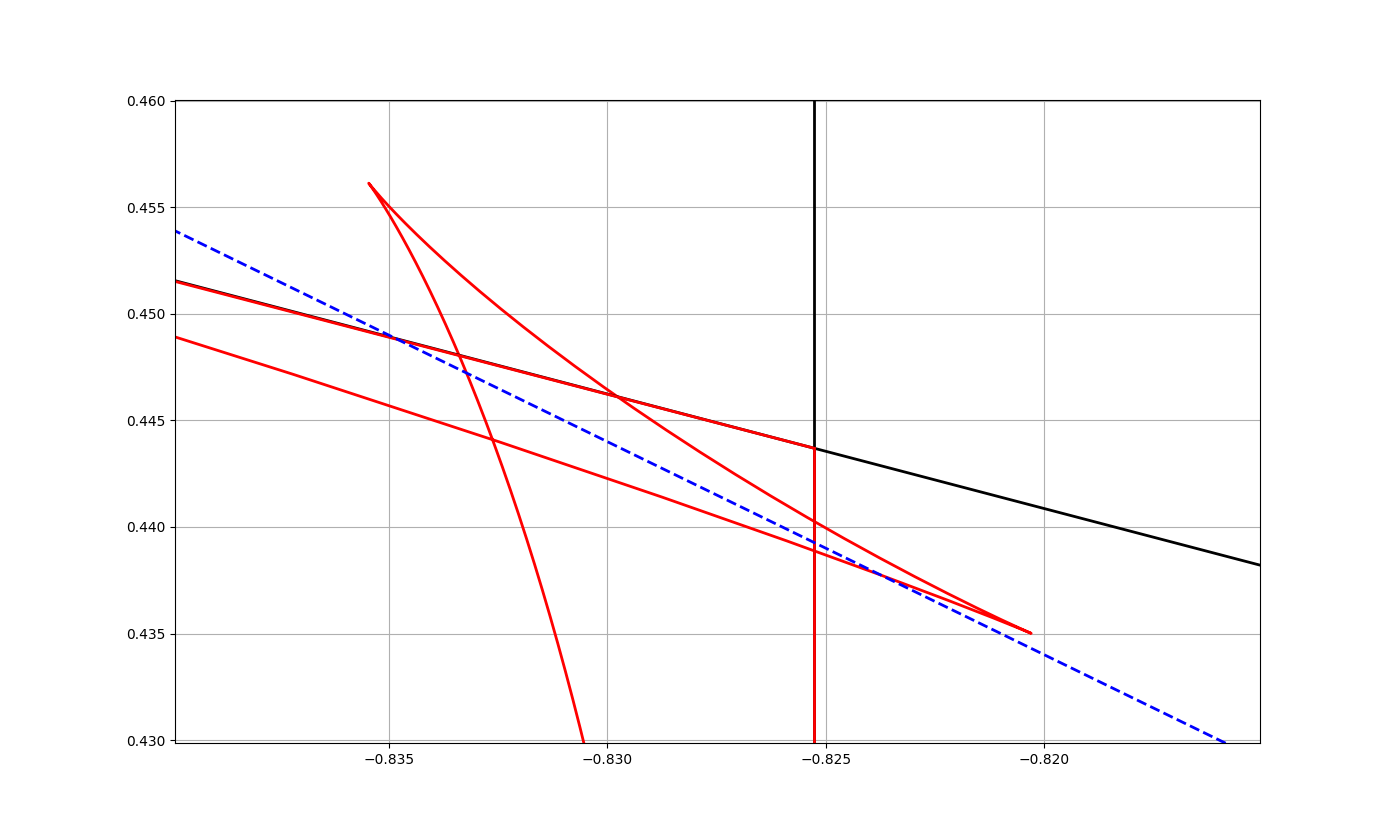}
                \caption{The blue line (dashed) corresponding to the chosen short rate $r= -0.386$ passes through 5 different regions of the state space. Parameter values are the same as in \cref{fig:sp} and $\kappa = 0.0$.}
                \label{fig:srl}
        \end{figure}
        \begin{figure}
                \centering
                \includegraphics[width=1\textwidth]{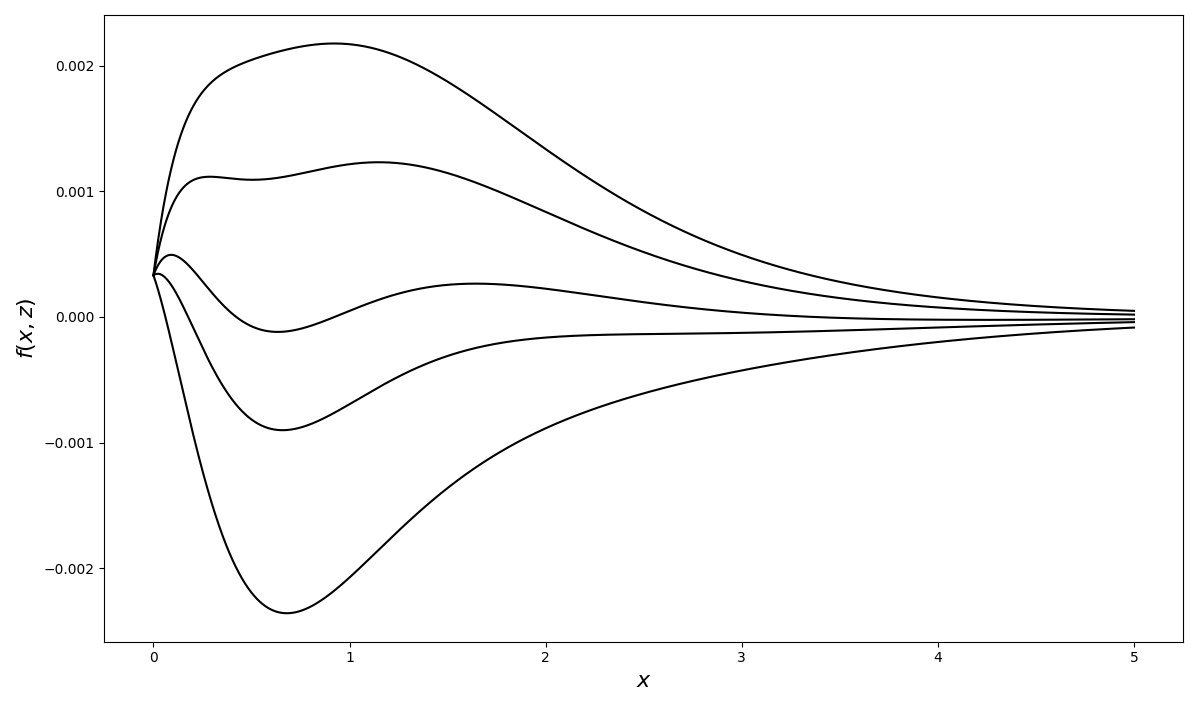}
                \caption{Different forward curves for the same short rate, with shapes from top to bottom: \texttt{h}, \texttt{hdh}, \texttt{hdhd}, \texttt{hd}, \texttt{d}}
                \label{fig:srf}
        \end{figure}
    
    \subsection{Computing probabilities of shapes}

 With the help of \cref{thm:main} we are able to tell which shape the forward curve has in a given point $z\in\RR^2$ of the state-space. We can use that to efficiently compute the probability of observing a certain shape (under both risk-neutral or physical measure) by Monte-Carlo simulation. For an example, we use parameters that were already used in \cite[Fig.3]{diez2020yield}, but change the value of $\sigma_2$ slightly for visualization purposes.
    \begin{align*}
        \rho = -0.996,\ \lambda_1 = 0.178,\ \lambda_2 = 0.401,\\ \sigma_1 = 0.0372,\ \sigma_2 = 0.037,\ \theta_1 = 0.0,\ \theta_2 = 0.01297.
    \end{align*}
    For this set of parameters the 5 different shapes $\texttt{n}$, $\texttt{h}$, $\texttt{d}$, $\texttt{i}$ and $\texttt{dh}$ occur. To determine the winding number of the augmented envelope around some sample point we use an axis-crossing algorithm from \cite{alciatore1995winding}, which is basically \cref{thm:non-zero} applied to polygons. We sample according to the risk-neutral stationary distribution of the Ornstein-Uhlenbeck process $(Z_t)_{t \ge 0}$, given our parameters. In \cref{fig:Samples}, 3000 samples, together with the state-space decomposition can be seen. Using 100.000 sample points we get the following estimated probabilities:
    \begin{align*}
        \texttt{n}: 0.09076,\ \texttt{h}: 0.4087,\ \texttt{d}: 0.25521,\ \texttt{i}: 0.2247,\ \texttt{dh}: 0.02063.
    \end{align*}
    For `rare' regions the efficiency of the Monte-Carlo estimation can be improved by using importance sampling.

    \begin{figure}
        \centering
        \includegraphics[width=1\textwidth]{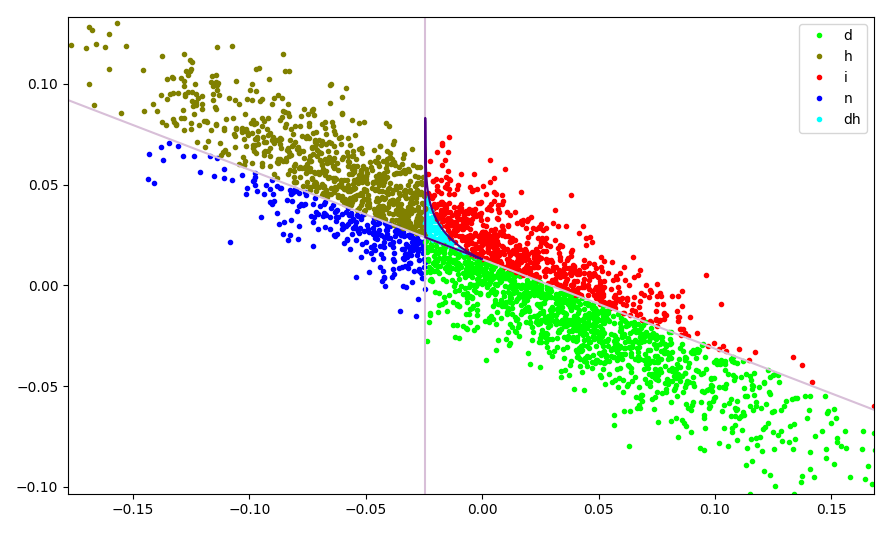}
        \caption{Samples from the risk-neutral stationary distribution, colored by shape of forward curve}
        \label{fig:Samples}
    \end{figure}
    
\section{Outlook}

We have introduced a new method, the method of envelopes, to determine all state-contingent shapes of the forward- and the yield-curve in the two-dimensional Vasicek model. Some minor questions related to the complete classification of shapes and transition graphs in the scale-proximal case with $\rho < 0$ remain open. In principle, our method could be applied to other two-dimensional affine term structure models (see \cite{dai2000specification}), but it remains to check conditions \ref{item:nonzero} -- \ref{item:oblique} from \cref{sec:def}, or -- if they are not satisfied -- to appropriately adapt the theory. Even non-affine models seem to be within the reach of the method, since the concept of envelopes can be applied to any family of differentiable curves, not just of lines, cf. \cite{bruce1992curves}. Finally, it remains open, whether (and how) our approach can be extended to dimensions higher than two, i.e., to interest rate models with three or more factors, when lines and curves have to be replaced by hyperplanes and hypersurfaces, and when no obvious analogue of a winding number exists. 
    
\appendix
\section{Additional results}
  The following result shows some important connections between the envelopes in the forward- and in the yield-curve case.

    \begin{lem}\label{lem:start}
    The following holds true for the envelopes $\eta_{\rm f}$ and $\eta_{\rm y}$ associated to the forward- and to the yield-curve respectively in the two-dimensional Vasicek model: 
    
    \begin{enumerate}[(a)]
    \item The envelopes $\eta_{\rm f}$ and $\eta_{\rm y}$ start in the same point, i.e., $\eta_{\rm f}(0) = \eta_{\rm y}(0)$;
    \item The endpoint of $\eta_{\rm y}$ is the intersection of $\ell_\infty^{\rm f}$ and $\ell_\infty^{\rm y}$;
    \item The number of zeros of $W(a_{\rm{y}}, b_{\rm{y}}, c_{\rm{y}})$ is less or equal than the number of zeros of $W(a_{\rm{f}}, b_{\rm{f}}, c_{\rm{f}})$;
    \item Every cusp point of $\eta_{\rm y}$ is also a point of $\eta_{\rm f}$.
    \end{enumerate}
    \end{lem}
    \begin{proof}
     By \cref{lem:envelope_slope}, the starting point of the envelope $\eta_{\rm f}$ can be written as
        \begin{equation}\label{eq:initial_again}
            \eta_{\rm f}(0) = \left(\frac{W(c_{\rm f}, a_{\rm f})(0)}{W(b_{\rm f}, c_{\rm f})(0)}, \frac{W(a_{\rm f}, b_{\rm f})(0)}{W(b_{\rm f}, c_{\rm f})(0)}\right).
        \end{equation}
    Let $g$ stand for any of $\set{a,b,c}$ and distinguish between $g_{\rm f}$ and $g_{\rm y}$ according to the case of forward- and yield-curve. From \eqref{eq:fy} it holds that 
        \begin{align*}
            g_{\rm y}(x) = \frac{1}{x^2}\int_0^x\zeta g_{\rm f}(\zeta)\dif\zeta
        \end{align*}
        for any choice of $g \in \set{a,b,c}$. 
        We are interested in the values of $g_{\rm y}(0)$ and $g_{\rm y}'(0)$. Using L'H\^opital's rule twice yields
        \begin{align*}
            g_{\rm y}(0) = \frac{g_{\rm f}(0)}{2}.
        \end{align*}
        Differentiating $g_{\rm y}$ yields
        \begin{align*}
            g_{\rm y}'(x) = \frac{1}{x}\left[g_{\rm f}(x) - \frac{2}{x^2}\int_0^x\zeta g_{\rm f}(\zeta)\dif\zeta\right]
        \end{align*}
        and using L'H\^opital's rule again we get
        \begin{align*}
            g_{\rm y}'(0) = \lim_{x\downarrow 0}\left(g_{\rm f}'(x) + \frac{2}{x}\left[\frac{2}{x^2}\int_0^x\zeta g_{\rm f}(\zeta)\dif\zeta - g_{\rm f}(x)\right]\right) = g_{\rm f}'(0) - 2g_{\rm y}'(0), 
        \end{align*}
and therefore $g_{\rm y}'(0) = \tfrac{g_{\rm f}'(0)}{3}$. Applying these results to $a,b$ and $c$ in \eqref{eq:initial_again} yields $\eta_{\rm y}(0) = \eta_{\rm f}(0)$, showing (a). 
Differentiating $g_{\rm y}$ once again, we obtain the identities
\begin{align*}
            g_{\rm y}'(x) &= \frac{1}{x} g_{\rm f}(x) - \frac{2}{x}g_{\rm y}(x)\\
            g_{\rm y}''(x) &= -\frac{3}{x^2}g_{\rm f}(x) + \frac{1}{x}g_{\rm f}'(x) + \frac{6}{x^2}g_{\rm y}(x).
        \end{align*}
Applying this to $a,b$ and $c$ yields
        \begin{align*}
            \begin{pmatrix} a_{\rm y}(x) & b_{\rm y}(x) & c_{\rm y}(x)\\
                            a_{\rm y}'(x) & b_{\rm y}'(x) & c_{\rm y}'(x)\\
                            a_{\rm y}''(x) & b_{\rm y}''(x) & c_{\rm y}''(x) \end{pmatrix}
            \begin{pmatrix} 1 \\ z_1 \\ z_2\end{pmatrix} = 0 \Leftrightarrow
            \begin{pmatrix} a_{\rm y}(x) & b_{\rm y}(x) & c_{\rm y}(x)\\
                            a_{\rm f}(x) & b_{\rm f}(x) & c_{\rm f}(x)\\
                            a_{\rm f}'(x) & b_{\rm f}'(x) & c_{\rm f}'(x) \end{pmatrix}
            \begin{pmatrix} 1 \\ z_1 \\ z_2\end{pmatrix} = 0.
        \end{align*}
If the left hand side is satisfied by a point $z = (z_1, z_2) = \eta_{\rm y}(x)$, then this point is a cusp point of $\eta_{\rm y}$ and it holds that $W(a_{\text{y}}, b_{\text{y}}, c_{\text{y}})(x) = 0$. The right hand side then implies that $z$ is also part of $\eta_{\rm f}$, i.e., $z = \eta_{\rm f}(x)$, showing (d).  Now suppose that $\eta_{\rm y}$ has more cusp points than $\eta_{\rm f}$. Then there must be times $0 \le x_1 < x_2$ (either the starting time or times of cusps of $\eta_{\rm y}$) such that $\eta_{\rm y}$ touches $\eta_{\rm f}$ at these points, and that $\eta_{\rm f}$ has \emph{no} cusp point in $(x_1, x_2)$. Cauchy's mean value theorem applied to the corresponding segments of $\eta_{\rm f}$ and $\eta_{\rm y}$ implies that their tangent vectors must be parallel at some $\xi \in (x_1, x_2)$; equivalently it must hold that 
        \begin{align*}
            \begin{vmatrix} b_{\rm f}(\xi) & c_{\rm f}(\xi) \\ b_{\rm y}(\xi) & c_{\rm y}(\xi)\end{vmatrix} = 0. 
        \end{align*}
Applying the identities from above this is equivalent to 
        \begin{align*}
            W(b_{\rm y}, c_{\rm y})(\xi) = 0, 
        \end{align*}
        which contradicts \ref{item:Wbc}, and completes the proof of (c). To show (b), note that the above identities imply that        \begin{align*}
            \begin{pmatrix} a_{\rm y}(x) & b_{\rm y}(x) & c_{\rm y}(x)\\
                            a_{\rm y}'(x) & b_{\rm y}'(x) & c_{\rm y}'(x) \end{pmatrix}
            \begin{pmatrix} 1 \\ z_1 \\ z_2\end{pmatrix} = 0 \Leftrightarrow
            \begin{pmatrix} a_{\rm y}(x) & b_{\rm y}(x) & c_{\rm y}(x)\\
                            a_{\rm f}(x) & b_{\rm f}(x) & c_{\rm f}(x) \end{pmatrix}
            \begin{pmatrix} 1 \\ z_1 \\ z_2\end{pmatrix} = 0.
        \end{align*}
        and (b) follows by taking the limit $x\rightarrow\infty$.
    \end{proof}

\bibliographystyle{alpha}
\bibliography{references}
\end{document}